\documentclass[english]{article}
\usepackage[T1]{fontenc}
\usepackage[latin9]{inputenc}
\usepackage{geometry}
\usepackage{array}
\usepackage{float}
\usepackage{mathtools}
\usepackage{url}
\usepackage{multirow}
\usepackage{amsmath}
\usepackage{amsthm}
\usepackage{amssymb}
\usepackage{graphicx}
\usepackage[authoryear,round]{natbib}

\makeatletter
\providecommand{\tabularnewline}{\\}

\floatstyle{ruled}
\newfloat{algorithm}{tbp}{loa}
\providecommand{\algorithmname}{Algorithm}
\floatname{algorithm}{\protect\algorithmname}
\theoremstyle{remark}
\newtheorem{rem}{\protect\remarkname}[section]
\theoremstyle{plain}
\newtheorem{thm}{\protect\theoremname}[section]
\theoremstyle{definition}
\newtheorem{defn}{\protect\definitionname}[section]
\theoremstyle{plain}
\newtheorem{lem}{\protect\lemmaname}[section]
\theoremstyle{plain}
\newtheorem{cor}{\protect\corollaryname}[section]

\usepackage{hyperref}
\usepackage{microtype}
\usepackage{algorithmic}
\usepackage{graphicx}
\usepackage{xcolor}
\date{}

\ifdefined\showcaptionsetup
 \PassOptionsToPackage{caption=false}{subfig}
\fi
\usepackage{subfig}
\makeatother

\usepackage{babel}
\providecommand{\corollaryname}{Corollary}
\providecommand{\definitionname}{Definition}
\providecommand{\lemmaname}{Lemma}
\providecommand{\remarkname}{Remark}
\providecommand{\theoremname}{Theorem}

\begin{document}
\global\long\def\R{\mathbb{\mathbb{R}}}%
\global\long\def\opt{{\cal OPT}}%
\global\long\def\res{\mathrm{res}}%
\global\long\def\ressolver{\mathrm{ResidualSolver}}%
\global\long\def\subsolver{\mathrm{SubSolver}}%
\allowdisplaybreaks
\title{Improved $\ell_{p}$ Regression via Iteratively Reweighted Least Squares}
\author{Alina Ene\thanks{Department of Computer Science, Boston University, \texttt{aene@bu.edu}.}\and 
Ta Duy Nguyen\thanks{Department of Computer Science, Boston University, \texttt{taduy@bu.edu}.}\and 
Adrian Vladu\thanks{CNRS \& IRIF, Universit\'e Paris Cit\'e, \texttt{vladu@irif.fr}.}}
\maketitle
\begin{abstract}
We introduce fast algorithms for solving $\ell_{p}$ regression problems
using the iteratively reweighted least squares (IRLS) method. Our
approach achieves state-of-the-art iteration complexity, outperforming
the IRLS algorithm by Adil-Peng-Sachdeva (NeurIPS 2019) and matching
the theoretical bounds established by the complex algorithm of Adil-Kyng-Peng-Sachdeva
(SODA 2019, J. ACM 2024) via a simpler lightweight iterative scheme.
This bridges the existing gap between theoretical and practical algorithms
for $\ell_{p}$ regression. Our algorithms depart from prior approaches,
using a primal-dual framework, in which the update rule can be naturally
derived from an invariant maintained for the dual objective. Empirically,
we show that our algorithms significantly outperform both the IRLS
algorithm by Adil-Peng-Sachdeva and MATLAB/CVX implementations.
\end{abstract}

\section{Introduction}

In this paper, we study the $\ell_{p}$ regression problem defined
as follows. The input to the problem is a matrix $A\in\R^{d\times n},$
a vector $b\in\R^{d}$ that lies in the column span of $A$, and an
accuracy parameter $\epsilon$. The goal is to approximately solve
the problem $\min_{x\in\R^{n}\colon Ax=b}\|x\|_{p}$, i.e., find a
solution $x\in\R^{n}$ such that $Ax=b$ and $\|x\|_{p}\leq(1+\epsilon)\|x^{*}\|_{p}$,
where $x^{*}$ is an optimal solution to the problem, and $\left\Vert \cdot\right\Vert _{p}$
denotes the $\ell_{p}$ norm. Solving $\ell_{p}$ regression for all
values of $p$ is a fundamental problem in machine learning with numerous
applications and has been studied in a long line of research beyond
the classical least squares regression with $p=2$.  $L_{p}$-norm
regression problems with general $p$ arise in several areas, including
supervised learning, graph clustering, and wireless networks. Examples
of applications include $\ell_{p}$-norm based algorithms in semi-supervised
learning \citep{DBLP:conf/colt/Alaoui16,DBLP:conf/nips/LiuG20}, $k$-clustering
with $\ell_{p}$-norm \citep{DBLP:conf/stoc/HuangV20}, robust regression
and robust clustering \citep{meng2013robust,DBLP:conf/iclr/HuangJL023}.

For this general class of convex optimization problems, designing
provably fast iterative algorithms to obtain high accuracy solutions
with empirical efficiency is an important question. General convex
programming methods such as interior point methods  are usually slow
in practice. In theory, \citet{bubeck2018homotopy} show that algorithms
based on interior point methods cannot improve beyond $O(\sqrt{n})$
iterations\footnote{For simplicity in the introduction, we assume that $d=\Theta(n)$.
In the regime when $n\gg d$, the IPM iteration complexity improves
to $\widetilde{O}(\sqrt{d})$. }  for any $p\notin\left\{ 1,2,\infty\right\} $. Breaking this barrier
and finding iterative algorithms that are faster than interior point
methods both in theory and practice is the goal of this line of work.

Recent developments have led to new algorithmic approaches such as
a homotopy method \citep{bubeck2018homotopy}, and an iterative refinement
approach \citep{adil2019iterative,adil2019fast,adil2022fast} for
$\ell_{p}$ regression with $p\notin\{1,\infty\}$. We highlight the
notable works by \cite{adil2019iterative,adil2019fast,adil2022fast}.
On the one hand, the algorithm with the best known theoretical runtime
is given by \citet{adil2019iterative,adil2022fast} with $O\big(p^{2}n^{\frac{p-2}{3p-2}}\log\big(\frac{n}{\epsilon}\big)\big)$
calls\footnote{The original result is $O\big(pn^{\frac{p-2}{3p-2}}\log\big(\frac{\left\Vert x^{(0)}\right\Vert _{p}^{p}-\left\Vert x^{*}\right\Vert _{p}^{p}}{\epsilon}\big)\big)$
for finding $\widehat{x}$ such that $\left\Vert \widehat{x}\right\Vert _{p}^{p}\le\min_{x:Ax=b}\left\Vert x\right\Vert _{p}^{p}+\epsilon$.
This translates to $O\big(pn^{\frac{p-2}{3p-2}}\log\big(\frac{\left\Vert x^{(0)}\right\Vert _{p}^{p}-\left\Vert x^{*}\right\Vert _{p}^{p}}{p\epsilon\left\Vert x^{*}\right\Vert _{p}^{p}}\big)\big)=O\big(p^{2}n^{\frac{p-2}{3p-2}}\log\big(\frac{n}{\epsilon}\big)\big)$
for finding $\widehat{x}$ such that $\left\Vert \widehat{x}\right\Vert _{p}\le\left(1+\epsilon\right)\min_{x:Ax=b}\left\Vert x\right\Vert _{p}$
for $x^{(0)}$ initialized to $\min_{x:Ax=b}\left\Vert x\right\Vert _{2}$.} to a linear system solver. This algorithm, however, relies on complex
subroutines and includes theoretical choices for several hyperparameters.
In practice, to obtain an efficient implementation, hyperparameters
require tuning. Due to these reasons, this theoretical algorithm by
\citet{adil2019iterative,adil2022fast} does not provide a practical
implementation. On the other hand, an algorithm known as $p$-IRLS
by \citet{adil2019fast} has been shown to have significant speed
up over standard solvers such as CVX. This algorithm is implemented
based on an Iteratively Reweighted Least Squares (IRLS) method, which
is a general iterative framework for solving regression problems.
The key element of an IRLS method is solving a weighted least squares
regression problem in each iteration. This is equivalent to solving
a linear system of the form $\min_{x\in\R^{n}:Ax=b}x^{\top}Rx$, where
$R$ is a diagonal matrix, which can be computed very efficiently
in practice with the advance of numerical solvers. IRLS algorithms
are favored in practice \citep{burrus2012iterative}, but designing
IRLS algorithms with strong convergence guarantees is challenging.
In particular, to obtain the efficiency, the algorithm by \citet{adil2019fast}
sacrifices the theoretical guarantee, requiring $O\big(p^{3}n^{\frac{p-2}{2p-2}}\log\big(\frac{n}{\epsilon}\big)\big)$
linear system solves. This brings forth the question:
\begin{center}
\emph{Can we design an algorithm that retains the empirical efficiency
of an IRLS approach while achieving the state-of-the-art theoretical
runtime?}
\par\end{center}

In this work, we give a positive answer to this question. We provide
a new algorithmic framework for $\ell_{p}$ regression based on an
IRLS approach for all values of $p\in(1,\infty)$. We propose an algorithm
that uses $O\big(p^{2}n^{\frac{p-2}{3p-2}}\log\left(\frac{n}{\epsilon}\right)\big)$
linear system solves, matching the state-of-the-art theoretical algorithm
by \citet{adil2019iterative}, and improving upon the guarantee of
$O\big(p^{3}n^{\frac{p-2}{2p-2}}\log\left(\frac{n}{\epsilon}\right)\big)$
for the $p$-IRLS algorithm by \citet{adil2019fast}. We experimentally
compare our algorithm with the $p$-IRLS algorithm \citep{adil2019fast}
and CVX solvers, and we observe significant improvements in all instances.

\subsection{Our contributions}

For the simplicity of the exposition, we study the $\ell_{p}$ regression
problem in both low and high precision regimes for $p\ge2$. 
\begin{rem}
In Appendix \ref{sec:Reducing-General-Regression}, we show a simple
reduction for the more general problem $\min_{x\colon Ax=b}\|Nx-v\|_{p}$
to the form $\min_{x\colon\tilde{A}x=\tilde{b}}\|x\|_{p}$ with the
dependence of the runtime on the number of rows of $N$ instead of
the dimension of $x$. We also show in Appendix \ref{sec:small-p-reduction}
a reduction for the case $1<p<2$ to the case $p\ge2$.
\end{rem}

In the low precision regime when the runtime dependence on $\epsilon$
is $\mathrm{poly}\left(\frac{1}{\epsilon}\right)$, we have the following
theorem.
\begin{thm}
\label{thm:low-precision}For any $p\geq2,$ there is an iterative
algorithm for the $\ell_{p}$ regression problem $\min_{x\in\R^{n}\colon Ax=b}\|x\|_{p}$
that solves $O\left(\log\log n+\log\left(1/\epsilon\right)\right)$
subproblems, each of which makes $O\Big(\Big((\frac{1}{\epsilon})^{\frac{2p-3}{p-2}}+n^{\frac{p-2}{3p-2}}(\frac{1}{\epsilon})^{\frac{3p^{2}-4p}{3p^{2}-8p+4}}\Big)\log\big(\frac{n}{\epsilon^{\frac{p}{p-2}}}\big)\Big)$
calls to solve a linear system of the form $ADA^{\top}\phi=b$, where
$D$ is an arbitrary non-negative diagonal matrix.

\end{thm}
\begin{rem}
When $p=\infty$, each subproblem makes $O\big(\frac{1}{\epsilon^{2}}+\frac{n^{\frac{1}{3}}}{\epsilon}\log(\frac{n}{\epsilon})\big)$
calls to a linear system solver. 
\end{rem}
Prior approaches for solving $\ell_{p}$ regression problem in the
low precision regime commonly use the Taylor expansion of $\|x\|_{p}^{p}$,
which then allows for deriving and bounding the updates. In contrast
to this, our algorithm relies on a primal-dual approach using the
dual formulation of the squared objective $\min_{x\colon Ax=b}\|x\|_{p}^{2}=\min_{x\colon Ax=b}\|x^{2}\|_{p/2}=\max_{r}\frac{\mathcal{E}(r)}{\left\Vert r\right\Vert _{q}}$
where $\ell_{q}$ is the dual norm of $\ell_{p/2}$ and $\mathcal{E}(r)=\min_{x:Ax=b}\langle r,x^{2}\rangle$.
The term $\mathcal{E}(r)$ is often referred to as the energy. The
high level idea of our approach is as follows. Starting with an initial
solution $r$ for the dual problem, we will increase the coordinates
of $r$ as much as possible so that the increase in the energy $\mathcal{E}(r)$
relative to the increase of $\left\Vert r\right\Vert _{q}$ is also
sufficiently large, until we can obtain a $(1-\epsilon)$ optimal
dual solution and whereby recover an approximately optimal primal
solution.  This template is close to the approach for $\ell_{\infty}$
regression by \citet{ene2019improved}. However, $\ell_{p}$ regression
does not have the readily decomposable structure along the coordinates
as $\ell_{\infty}$ regression and novel technique is required in
the design of the algorithm. Our approach is also a reminiscence of
the width-independent multiplicative weights update method for solving
mixed packing covering linear program, where in each step the algorithm
updates the coordinates the maximize the bang-for-buck ratio \citep{quanrud2020nearly}.
In contrast to MWU, we do not use a mirror map or regularize $\ell_{p}$
norms to make them smooth as in standard approaches. Our scheme allows
our method to take much longer steps, where in each step, the coordinates
of the dual solution are allowed to change by large polynomial factors
and thereby achieve faster running time.

To obtain faster algorithms in the high accuracy regime with a logarithmic
dependence on the accuracy, we adapt the iterative refinement approach
of \citet{adil2019iterative} and obtain improved running times.
\begin{thm}
\label{thm:high-precision}For any $p\geq2,$ there is an iterative
algorithm for the $\ell_{p}$ regression problem $\min_{x\in\R^{n}\colon Ax=b}\|x\|_{p}$
that solves $O\left(p^{2}\log n\log\left(\frac{n}{\epsilon}\right)\right)$
subproblems, each of which makes $O\big(n^{\frac{p-2}{3p-2}}\big)$
calls  to solve a linear system of the form $\widetilde{A}D\widetilde{A}^{\top}\phi=z$,
where $D$ is an arbitrary non-negative diagonal matrix, $\widetilde{A}$
is a matrix obtained from $A$ by appending a single row, and $z$
is a vector obtained from the all-zero vector by appending a single
non-zero coordinate.  
\end{thm}
Using the iterative refinement template by \citep{adil2019iterative,adil2019fast,adil2022fast},
we instead use an IRLS solver for the residual problems with improved
runtime. The residual solver solves a mixed $\ell_{p}+\ell_{2}$ problem
in the form $\min_{x\colon Ax=b}\|x\|_{p}^{2}+\left\langle \theta,x^{2}\right\rangle $,
only to a constant approximation.  Here the challenge lies in the
fact that the $\ell_{2}$ term makes the dual problem no longer scale-free
and thus our low precision solver is not immediately usable.  However,
by an appropriate initialization of the dual solution and careful
adjustments to the step size, our algorithm achieves the desired $O\big(n^{\frac{p-2}{3p-2}}\big)$
bound. Since regularized $\ell_{p}+\ell_{2}$ regression problems
arise in many applications in machine learning and beyond, our algorithm
for the mixed $\ell_{p}+\ell_{2}$ objective is of independent interest.

Finally, we experimentally evaluate our high-precision algorithm.
Our algorithm significantly outperforms the $p$-IRLS algorithm \citep{adil2019iterative}
both in the number of linear system solves as well as the overall
running time. Our algorithm is significantly faster than CVX solvers
 and is able to run on large instances, which is not possible for
CVX solvers within a time constraint.

\subsection{Related work}

$\ell_{p}$ regression problems have received significant attention.
Here we summarize the results that are closest to our work. The surveyed
algorithms are iterative algorithms where the running time of each
iteration is dominated by a single linear system solve. 

Algorithms based on interior point methods use  $\widetilde{O}\left(\sqrt{n}\right)$
iterations for any $p\in[1,\infty]$ \citep{nesterov1994interior},
which was improved to $\widetilde{O}\left(\sqrt{d}\right)$ iterations
for $p\in\left\{ 1,\infty\right\} $ \citep{lee2014path}. Bubeck-Cohen-Lee-Li
\citep{bubeck2018homotopy} show that this iteration bound is generally
necessary for interior point methods and propose a homotopy-based
algorithm that uses $\tilde{O}\big(\mathrm{poly}\big(\frac{p^{2}}{p-1}\big)\cdot n^{|1/2-1/p|}\big)$
iterations for any $p\notin\{1,\infty\}$. \cite{adil2019iterative,adil2022fast}
introduced an iterative refinement framework that uses $O\big(p^{2}\cdot n^{\frac{p-2}{3p-2}}\log(\frac{n}{\epsilon})\big)$
iterations for any $p>2$. Using Lewis weight sampling, Jambulapati-Liu-Sidford
\citep{jambulapati2022improved} improve the method by \citet{adil2019iterative,adil2022fast}
to $O\big(p^{p}\cdot d^{\frac{p-2}{3p-2}}\mathrm{polylog}(\frac{n}{\epsilon})\big)$,
for overconstrained regression problems $\min_{x\in\R^{d}}\left\Vert Ax-b\right\Vert _{p}$
where $A\in\R^{n\times d}$ and $n$ is much larger than $d$ (the
iteration complexity of the prior algorithms will still depend on
the larger dimension $n$ in this case). \citet{bullins2018fast}
gives a faster algorithm for minimizing structured convex quartics,
which implies an algorithm for $\ell_{4}$ regression with $\tilde{O}(n^{\frac{1}{5}})$
iterations. Building on the work of \citet{ChristianoKMST11,ChinMMP13}
for maximum flows and regression, \citet{ene2019improved} give an
algorithm for $\ell_{1}$ and $\ell_{\infty}$ regression using $O\big(\frac{n^{1/3}\log(1/\epsilon)}{\epsilon^{2/3}}+\frac{\log n}{\epsilon^{2}}\big)$
iterations. This work also uses a primal-dual framework but the algorithm
and analysis are specific to the special structure of the $\ell_{1}$
and $\ell_{\infty}$ norm and work only in the low precision regime
with $\mathrm{poly}(\frac{1}{\epsilon})$ convergence.

\section{Our Algorithm with $\mathrm{poly}\left(\frac{1}{\epsilon}\right)$
Convergence\protect\label{sec:Solving-to-Low-Precision}}

In this section, we present our algorithm with guarantee provided
in Theorem \ref{thm:low-precision}.

Before describing the algorithm, we first introduce some basic notations.
For a constant $a\in\R$, we abuse the notation and use $a\in\R^{n}$
to denote the vector with all entries equal to $a$ (the dimension
will be clear from context). When it is clear from the context, we
apply scalar operations to vectors with the interpretation that they
are applied coordinate-wise. For $p\ge1$, we let $q$ be such that
$\frac{1}{p}+\frac{1}{q}=1$ and $\ell_{q}$ is the dual norm of the
$\ell_{p}$ norm.

\subsection{Our Algorithm}

For ease of notation, it is convenient to consider the following equivalent
formulation of the problem: For $p\ge1$, we solve $\min_{x:Ax=b}\left\Vert x\right\Vert _{2p}^{2}=\min_{x:Ax=b}\left\Vert x^{2}\right\Vert _{p}$
to $(1+\epsilon)$ multiplicative error. We provide our algorithm
in Algorithms \ref{alg:main-low} and \ref{Alg:low-precision}. We
give an overview of our approach and explain the intuition in the
following section.

\begin{algorithm}[t]
\caption{$\ell_{2p}$-minimization$\left(A,b,\epsilon\right)$}

\label{alg:main-low}

\textbf{Input: }Matrix $A\in\R^{d\times n}$, vector $b\in\R^{d}$,
accuracy $\epsilon$

\textbf{Output}: Vector $x$ such that $Ax=b$ and $\left\Vert x\right\Vert _{2p}\leq(1+\epsilon)\min_{x:Ax=b}\left\Vert x\right\Vert _{2p}$

Initialize $x^{(0)}=\min_{x:Ax=b}\left\Vert x\right\Vert _{2}$

$L=\max\left\{ i:(1+\epsilon)^{i}\le\frac{\left\Vert x^{(0)}\right\Vert _{2}}{n^{\frac{1}{2}-\frac{1}{2p}}}\right\} $;
$U=\min\left\{ i:(1+\epsilon)^{i}\ge\left\Vert x^{(0)}\right\Vert _{2}\right\} $

\textbf{while} $L<U$:

$\qquad$$P=\lfloor\frac{L+U}{2}\rfloor$, $M=(1+\epsilon)^{P}$

$\qquad$\textbf{if} $\subsolver(A,b,\epsilon,M)$ is infeasible \textbf{then}

$\qquad\qquad$$L=P+1$

$\qquad$\textbf{else}

$\qquad\qquad$Let $x^{(t+1)}$ be the output of $\subsolver(A,b,\epsilon,M)$

$\qquad\qquad$$U=P$; $t\gets t+1$

$\qquad$\textbf{end if}

\textbf{end while}

\textbf{return} $x^{(t)}$
\end{algorithm}

\begin{algorithm}[t]
\caption{$\protect\subsolver(A,b,\epsilon,M)$}

\label{Alg:low-precision}

\textbf{Input: }Matrix $A\in\R^{d\times n}$, vector $b\in\R^{d}$,
accuracy $\epsilon$, target value $M$

\textbf{Output}: Vector $x$ such that $Ax=b$ and $\left\Vert x\right\Vert _{2p}\leq(1+\epsilon)M$,

$\quad\quad\quad$or approximate infeasibility certificate $r$, $\left\Vert r\right\Vert _{q}=1$.

$t=0$, $r^{(0)}=\frac{1}{n^{1/q}}$, $t'=0$, $s^{(t')}=0$

\textbf{while} $\left\Vert \left(r^{(t)}\right)\right\Vert _{q}\leq\frac{1}{\epsilon}$

$\qquad x^{(t)}=\arg\min_{x:Ax=b}\langle r^{(t)},x^{2}\rangle$

$\qquad$$\gamma_{i}^{(t)}=\begin{cases}
\frac{x_{i}^{2}\left\Vert r\right\Vert _{q}^{q-1}}{M^{2}r_{i}^{q-1}} & \text{if \ensuremath{\frac{x_{i}^{2}\left\Vert r\right\Vert _{q}^{q-1}}{r_{i}^{q-1}}\geq(1+\epsilon)M^{2}}}\\
1 & \text{otherwise }
\end{cases}$, for all $i$

$\qquad$\textbf{if} $\gamma^{(t)}=1$ \textbf{then return} $x^{(t)}$\textbf{
end if\hfill{}}$\triangleright$ \emph{Case 1}

$\qquad$$\alpha^{(t)}=\left(\gamma^{(t)}\right)^{\frac{1}{q}}$;
$r^{(t+1)}=r^{(t)}\cdot\alpha^{(t)}$

\textbf{$\qquad$if $\alpha^{(t)}\leq n^{\frac{2}{2q+1}}\left(\frac{1}{\epsilon}\right)^{\frac{q-1}{2q+1}}$
then $s^{(t'+1)}=s^{(t')}+x^{(t)}$};\textbf{ $t'=t'+1$ end if}

\textbf{$\qquad$if $t'>0$ and $\left\Vert s^{(t')}/t'\right\Vert _{2p}\leq(1+\epsilon)M$
then return $s^{(t')}/t'$ end if \hfill{}}$\triangleright$ \emph{Case
2}

$\qquad$$t=t+1$

\textbf{end while}

\textbf{return} $r^{(t)}$\textbf{\hfill{}}$\triangleright$ \emph{Case
3}
\end{algorithm}

\subsection{Overview of our approach}

Our algorithm is based on a primal-dual approach, starting with the
following dual formulation of the problem. Using $q$ as the dual
norm of $p$ and by duality, we write
\begin{align*}
\min_{x:Ax=b}\left\Vert x\right\Vert _{2p} & =\min_{x:Ax=b}\left\Vert x^{2}\right\Vert _{p}=\min_{x:Ax=b}\max_{r:\left\Vert r\right\Vert _{q}\leq1}\langle r,x^{2}\rangle\max_{r\ge0:\left\Vert r\right\Vert _{q}\leq1}\min_{x:Ax=b}\langle r,x^{2}\rangle=\max_{r\ge0}\frac{\mathcal{E}(r)}{\left\Vert r\right\Vert _{q}},
\end{align*}
where we defined $\mathcal{E}(r):=\min_{x:Ax=b}\langle r,x^{2}\rangle$.
The main part of our algorithm is the subroutine shown in Algorithm
\ref{Alg:low-precision}, which takes as input a guess $M$ for the
optimum value $\left\Vert x^{*}\right\Vert _{2p}$. To find an $(1+\epsilon)$
approximation of the optimum value, the main Algorithm \ref{alg:main-low}
performs a binary search as follows. Since $x^{(0)}$ is initialized
to $\min_{x:Ax=b}\left\Vert x\right\Vert _{2}$, we can show that
$\left\Vert x^{*}\right\Vert _{p}$ is contained in the range $\left[\frac{\left\Vert x^{(0)}\right\Vert _{2}}{n^{\frac{1}{2}-\frac{1}{2p}}},\left\Vert x^{(0)}\right\Vert _{2}\right]$.
The algorithm performs binary search over the indices $i$ such that
$\left(1+\epsilon\right)^{i}$ is in that range. Note that the main
algorithm only needs to perform at most $\log\left(\frac{\log n}{\epsilon}\right)$
iterations, each of which makes one call to the subproblem solver.

We now focus on the subproblem when we are given a guess $M$ and
a target precision $\epsilon$. The goal is to find a primal solution
$x$ that satisfies $\left\Vert x\right\Vert _{2p}\le M(1+\epsilon)$
or a dual solution $r$ (infeasibility certificate) which can certify
that $\min_{x:Ax=b}\left\Vert x\right\Vert _{2p}^{2}\ge\frac{\mathcal{E}(r)}{\left\Vert r\right\Vert _{q}}\ge(\frac{M}{1+\epsilon})^{2}$.
This lower bound on the optimal value of the problem tells us that
we can increase the guess $M$.

The objective function $\mathcal{E}(r)$ has a very useful monotonicity
property: it increases when $r$ increases. The overall strategy of
our algorithm is to start with an initial dual solution $r^{(0)}$
(which we initialize uniformly to $\frac{1}{n^{1/q}}$) and increase
it while maintaining the following invariant
\begin{align}
\mathcal{E}(r^{(t+1)})-\mathcal{E}(r^{(t)}) & \ge M^{2}(\left\Vert r^{(t+1)}\right\Vert _{q}-\left\Vert r^{(t)}\right\Vert _{q}),\label{eq:invar}
\end{align}
or equivalently,
\begin{align*}
\frac{\mathcal{E}(r^{(t+1)})-\mathcal{E}(r^{(t)})}{\left\Vert r^{(t+1)}\right\Vert _{q}-\left\Vert r^{(t)}\right\Vert _{q}} & \ge M^{2}.
\end{align*}
The telescoping property of both sides of (\ref{eq:invar}) will guarantee
that, if the algorithm outputs a dual solution $r$ with sufficiently
large $\left\Vert r\right\Vert _{q}$, this solution will satisfy
$\mathcal{E}(r)\ge\left(\frac{M}{1+\epsilon}\right)^{2}\left\Vert r\right\Vert _{q}$,
i.e, $\frac{\mathcal{E}(r)}{\left\Vert r\right\Vert _{q}}\ge\left(\frac{M}{1+\epsilon}\right)^{2}$.
To maintain the invariant \ref{eq:invar}, we have two useful bounds
for the change in the objective and dual solution:
\begin{align}
\mathcal{E}(r^{(t+1)})-\mathcal{E}(r^{(t)}) & \geq\sum_{i}r_{i}^{(t)}\left(x_{i}^{(t)}\right)^{2}\left(1-\frac{r_{i}^{(t)}}{r_{i}^{(t+1)}}\right),\label{eq:energy-increase}\\
\frac{1}{\left\Vert r^{(t+1)}\right\Vert _{q}-\left\Vert r^{(t)}\right\Vert _{q}} & \ge\frac{q\left\Vert r^{(t)}\right\Vert _{q}^{q-1}}{\sum_{i}\left(r_{i}^{(t+1)}\right)^{q}-\left(r_{i}^{(t)}\right)^{q}}.\label{eq:resistance-increase}
\end{align}
Both inequalities allow us to decompose the invariant along the coordinates.
That is, we can maintain the invariant by ensuring for each coordinate
$i$ that we increase that 
\begin{align*}
\frac{q\left\Vert r^{(t)}\right\Vert _{q}^{q-1}r_{i}^{(t)}\left(x_{i}^{(t)}\right)^{2}}{\left(r_{i}^{(t+1)}\right)^{q}-\left(r_{i}^{(t)}\right)^{q}}\left(1-\frac{r_{i}^{(t)}}{r_{i}^{(t+1)}}\right) & \ge M^{2}.
\end{align*}
In order to do this, we update each $r_{i}^{(t)}$ multiplicatively,
via the term $\gamma_{i}^{(t)}=\frac{\left\Vert r^{(t)}\right\Vert _{q}^{q-1}}{\left(r_{i}^{(t)}\right)^{q-1}}\cdot\frac{\left(x_{i}^{(t)}\right)^{2}}{M^{2}}$.
To guarantee fast convergence, we want to increase $r_{i}^{(t)}$
as much as possible, by setting a target threshold on $\gamma_{i}^{(t)}$:
if $\gamma_{i}^{(t)}$ exceeds the threshold, we update $r_{i}^{(t+1)}=r_{i}^{(t)}\left(\gamma_{i}^{(t)}\right)^{1/q}$;
otherwise, $r_{i}^{(t)}$ remains unchanged. When we can no longer
increase $r$ while preserving the invariant, we can be sure that
we have found the corresponding primal solution $x$ with small norm.
During the course of the algorithm, we also keep track of iterations
with small increases in $r$ and use the uniform average over the
corresponding primal solutions to obtain an approximately feasible
primal solution, in case the algorithm fails to return an infeasibility
certificate quickly enough.

We note that our update approach is derived in a completely different
way from standard iterative frameworks such as multiplicatives weights
updates and, generally, mirror descent. In contrast to these standard
approaches, we do not use a mirror map or regularize $\ell_{p}$ norms
to make them smooth. Our update scheme allows our algorithm to take
much longer steps, and the coordinates of the dual solution are allowed
to change by large polynomial factors in each step. This allows us
to obtain a fast convergence rate. 

We provide the complete analysis and proof of Theorem \ref{thm:low-precision}
in Appendix \ref{sec:Proof-of-Theorem-1}.

\section{Our Algorithm with $\log\left(\frac{1}{\epsilon}\right)$ Convergence\protect\label{sec:Solving-to-High-Precision}}

\subsection{Algorithm}

In this section, we present our algorithm with guarantee provided
in Theorem \ref{thm:high-precision}. For the ease of the exposition,
we consider a slight variation of the problem: for $p\ge2$, we solve
$\min_{x:Ax=b}\left\Vert x\right\Vert _{p}^{p}$ to $(1+\epsilon)$
multiplicative error. We show our algorithm in Algorithms \ref{alg:main-algo}
and \ref{Alg:regularized-regression-1}. 

\begin{algorithm}[t]
\caption{Iteratively Reweighted Least Squares}

\label{alg:main-algo}

\textbf{Input: }Matrix $A\in\R^{d\times n}$, vector $b\in\R^{d}$,
$\epsilon$

\textbf{Output}: Vector $x$ such that $Ax=b$ that minimizes $\left\Vert x\right\Vert _{p}^{p}$

Initialize $x^{(0)}=\arg\min_{x:Ax=b}\left\Vert x\right\Vert _{2}^{2}$

$M^{(0)}:=\frac{\left\Vert x^{(0)}\right\Vert _{p}^{p}}{16p}$, $t\gets0$;
$\kappa=\begin{cases}
1 & \text{if }p\le\frac{2\log n}{\log n-1}\\
\frac{p}{p-2} & \text{otherwise}
\end{cases}$

\textbf{while} $M^{(t)}\ge\frac{\epsilon}{16p\left(1+\epsilon\right)}\left\Vert x^{(t)}\right\Vert _{p}^{p}$

$\qquad$$g^{(t)}=\left|x^{(t)}\right|^{p-2}x^{(t)}$; $R^{(t)}=2\left|x^{(t)}\right|^{p-2}$

$\qquad$$\tilde{\Delta}\gets\ressolver\Bigg(\frac{p}{2},\left[\begin{array}{c}
A\\
(g^{(t)})^{\top}
\end{array}\right],$ $\left[0,\frac{M^{(t)}}{2}\right],(M^{(t)})^{\frac{2-p}{p}}R^{(t)},2\sqrt{\kappa}(M^{(t)})^{\frac{1}{p}}\Bigg)$

$\qquad$\textbf{if} $\tilde{\Delta}$ is an infeasibility certificate
or $\left\langle R^{(t)},\tilde{\Delta}^{2}\right\rangle \ge2M^{(t)}$
\textbf{then}

$\qquad\qquad$$M^{(t+1)}\gets M^{(t)}/2$, $x^{(t+1)}=x^{(t)}$

$\qquad$\textbf{else}

$\qquad\qquad$$M^{(t+1)}\gets M^{(t)}$, $x^{(t+1)}=x^{(t)}-\frac{\tilde{\Delta}}{64p\kappa}$

$\qquad$\textbf{end if}

$\qquad$$t\gets t+1$

\textbf{end while}

\textbf{return} $x^{(t)}$
\end{algorithm}

\begin{algorithm}[!th]
\caption{$\protect\ressolver(p,A,b,\theta,M)$}
\label{Alg:regularized-regression-1}

\textbf{Input: }Matrix $A\in\R^{d\times n}$, vector $b\in\R^{d}$,
target value $M$, weight $\theta$

\textbf{Output}: Vector $x$ such that $Ax=b$, $\left\Vert x\right\Vert _{2p}\leq2M$
and $\left\langle \theta,x^{2}\right\rangle \le\min_{x:Ax=b}\left\Vert x^{2}\right\Vert _{p}+\left\langle \theta,x^{2}\right\rangle $

$\quad\quad\quad$or approximate infeasibility certificate $r$, $\left\Vert r\right\Vert _{q}=1$.

\textbf{if} $p\le\frac{\log n}{\log n-1}$ \textbf{then}

$\qquad$$r=\frac{1}{n^{\frac{1}{q}}}$; $\widehat{x}=\arg\min_{x:Ax=b}\langle r+\theta,x^{2}\rangle$

$\qquad$\textbf{if} $\left\Vert \widehat{x}\right\Vert _{2p}\leq2M$
\textbf{then} return $\widehat{x}$\textbf{ else} \textbf{return}
$r$\textbf{ end if}

\textbf{else}

$\qquad$$t=0$, $r^{(0)}=\frac{2q-1}{2qn^{\frac{1}{q}}}$, $t'=0$,
$s^{(t')}=0$

$\qquad$\textbf{while} $\left\Vert \left(r^{(t)}\right)^{q}\right\Vert _{1}\leq1$

$\qquad\qquad x^{(t)}=\arg\min_{x:Ax=b}\langle r^{(t)}+\theta,x^{2}\rangle$

$\qquad$$\qquad\gamma_{i}^{(t)}=\begin{cases}
\frac{x_{i}^{2}\left\Vert r\right\Vert _{q}^{q-1}}{M^{2}r_{i}^{q-1}} & \text{if \ensuremath{\frac{x_{i}^{2}\left\Vert r\right\Vert _{q}^{q-1}}{r_{i}^{q-1}}\geq2M^{2}}}\\
1 & \text{otherwise }
\end{cases}$, for all $i$

$\qquad$$\qquad\alpha_{i}^{(t)}=\left(\gamma_{i}^{(t)}\right)^{1/q}$

$\qquad$$\qquad$\textbf{if} $\alpha^{(t)}=1$ \textbf{then return}
$x^{(t)}$\textbf{ end if\hfill{}}\emph{$\triangleright$ Case 1}

$\qquad\qquad$$r^{(t+1)}=\alpha^{(t)}\cdot r^{(t)}$

\textbf{$\qquad\qquad$if $\alpha^{(t)}\leq n^{\frac{2}{2q+1}}$ then
$s^{(t'+1)}=s^{(t')}+x^{(t)}$};\textbf{ $t'=t'+1$ end if}

\textbf{$\qquad\qquad$if $t'>0$ and $\left\Vert s^{(t')}/t'\right\Vert _{2p}\leq2M$
then return $s^{(t')}/t'$ end if \hfill{}}\emph{$\triangleright$
Case 2}

$\qquad\qquad$$t=t+1$

$\qquad$\textbf{end while}

\textbf{end if}

\textbf{return} $r^{(t)}$\textbf{\hfill{}}\emph{$\triangleright$
Case 3}
\end{algorithm}

\subsection{Overview of our approach}

At the highest level, the main algorithm relies on a simple yet powerful
observation by \citet{adil2019iterative}, which is that the $\ell_{p}$
minimization problem we are attempting to solve supports iterative
refinement. \citet{adil2019iterative} show that having access to
a weak solver which gives a constant factor multiplicative approximation
to a mixed objective of $\ell_{p}$ and $\ell_{2}$ norms suffices
to boost the multiplicative error to $1+\epsilon$ while only making
$\widetilde{O}_{p}(\log1/\epsilon)$ calls to the solver. This reduces
the entire difficulty of the problem to implementing the weak solver. 

More precisely, starting with an initial solution (set to $\arg\min_{x:Ax=b}\left\Vert x\right\Vert _{2}$),
we maintain $M^{(t)}$ as an upper bound for the function value gap,
ie. $\left\Vert x^{(t)}\right\Vert _{p}^{p}-\left\Vert x^{*}\right\Vert _{p}^{p}\le16pM^{(t)}$.
We show this invariant in Lemma \ref{lem:high-prec-invariant}. In
each iteration, the algorithm makes a call to a solver for the residual
problem which approximates the function value progress $\left\Vert x\right\Vert _{p}^{p}-\left\Vert x-\Delta\right\Vert _{p}^{p}$
if we update the solution $x\gets x-\Delta$. The residual solution
tells us either the progress is too small, in which case we can improve
the upperbound on the suboptimality gap by reducing $M^{(t)}$, or
the progress is at least $\Omega\left(M^{(t)}\right)$, in which case
we can perform the update and obtain a new solution. This new solution
improves the function value gap by at least a factor $1-\Omega\left(\frac{1}{p}\right)$,
and thus the algorithm requires only $O\left(p\log\frac{\left\Vert x^{(0)}\right\Vert _{p}^{p}-\left\Vert x^{*}\right\Vert _{p}^{p}}{\epsilon\left\Vert x^{*}\right\Vert _{p}^{p}}\right)$
calls to the residual solver. We show this guarantee in Lemma \ref{lem:high-prec-invariant}.

We give the pseudocode for the residual solver in Algorithm \ref{Alg:regularized-regression-1}\footnote{Note that while the residual solver takes as input the original matrix
$A$ augmented with an extra row, the least squares problems required
by the residual solver reduce to least squares problems involving
only $A$, using the Sherman-Morrison formula. This guarantees that
we only require a linear system solver for structured matrices of
the form $A^{\top}DA$, for non-negative diagonal $D$.}. Prior works by \citet{adil2019iterative,adil2019fast,adil2022fast}
give algorithms for this solver either via a width-reduced multiplicative
weights update method which achieves the state-of-the-art theoretical
runtime but does not support a practical implementation or via a practical
IRLS method with suboptimal theoretical guarantee. In contrast, we
build on ideas from the low precision IRLS solver we have shown in
the previous section and design a new IRLS algorithm that attains
the best of both worlds.

Our residual solver outputs an approximate solution to a constant
factor to the objective of the form 
\begin{equation}
\min_{x:Ax=b}\left\Vert x^{2}\right\Vert _{p}+\left\langle \theta,x^{2}\right\rangle \label{eq:obj}
\end{equation}
for $p\ge1$ and a positive weight vector $\theta\in\R^{n}$. We also
start with the dual formulation of the problem
\begin{align*}
\mbox{(\ref{eq:obj})}= & \min_{x:Ax=b}\max_{r:\left\Vert r\right\Vert _{q}=1}\left\langle r,x^{2}\right\rangle +\left\langle \theta,x^{2}\right\rangle = & \max_{r\ge0:\left\Vert r\right\Vert _{q}=1}\min_{x:Ax=b}\left\langle r+\theta,x^{2}\right\rangle = & \max_{r\ge0}\mathcal{E}\left(\frac{r}{\left\Vert r\right\Vert _{q}}+\theta\right),
\end{align*}
where $q$ is the dual norm to $p$ and $\mathcal{E}(r+\theta)=\min_{x:Ax=b}\langle r+\theta,x^{2}\rangle$.
Given a target $M$, our goal is to find a primal solution $x$ that
satisfies $\left\Vert x\right\Vert _{2p}\leq2M$ and $\left\langle \theta,x^{2}\right\rangle \le\min_{x:Ax=b}\left\Vert x^{2}\right\Vert _{p}+\left\langle \theta,x^{2}\right\rangle $
or a dual solution $r\in\R^{n}$ (infeasibility certificate) which
can certify that $\min_{x:Ax=b}\left\Vert x\right\Vert _{2p}^{2}\ge\mathcal{E}\big(\frac{r}{\left\Vert r\right\Vert _{q}}+\theta\big)\ge\frac{M^{2}}{2\kappa}$,
where $\kappa$  is a value set as shown in Algorithm \ref{alg:main-algo}.

We distinguish between two regimes: when $p$ is sufficiently small,
$1\le p\le\frac{\log n}{\log n-1}$ for which we will show that we
can obtain a solution by $O(1)$ calls to the linear solver, and when
$p>\frac{\log n}{\log n-1}$, to which we need to pay more attention.
In the latter case, similarly to Algorithm \ref{Alg:low-precision},
we want to maintain the invariant 
\begin{align*}
\frac{\mathcal{E}(r^{(t+1)}+\theta)-\mathcal{E}(r^{(t)}+\theta)}{\left\Vert r^{(t+1)}\right\Vert _{q}-\left\Vert r^{(t)}\right\Vert _{q}} & \ge M^{2}.
\end{align*}
Notice the differences between this objective and the problem $\min_{x:Ax=b}\left\Vert x^{2}\right\Vert _{p}$
which we solve in the previous section. The $\ell_{2}$ term $\left\langle \theta,x^{2}\right\rangle $
makes this objective no longer scale-free. However, this $\ell_{2}$
term does not affect the lower bound $\sum_{i}r_{i}^{(t)}\left(x_{i}^{(t)}\right)^{2}\left(1-\frac{r_{i}^{(t)}}{r_{i}^{(t+1)}}\right)$
in the change in the objective (eq. (\ref{eq:energy-increase}));
thus it suffices to maintain $\frac{\sum_{i}r_{i}^{(t)}\left(x_{i}^{(t)}\right)^{2}\left(1-\frac{r_{i}^{(t)}}{r_{i}^{(t+1)}}\right)}{\left\Vert r^{(t+1)}\right\Vert _{q}-\left\Vert r^{(t)}\right\Vert _{q}}\ge M^{2}$
in order to guarantee the invariant $\frac{\mathcal{E}(r^{(t+1)}+\theta)-\mathcal{E}(r^{(t)}+\theta)}{\left\Vert r^{(t+1)}\right\Vert _{q}-\left\Vert r^{(t)}\right\Vert _{q}}\ge M^{2}$.
At the same time, if we maintain $\left\Vert r\right\Vert _{q}\le1$,
we can show that if the algorithm outputs a primal solution $x$,
the $\ell_{2}$ term $\left\langle \theta,x^{2}\right\rangle \le\min_{x:Ax=b}\left\Vert x^{2}\right\Vert _{p}+\left\langle \theta,x^{2}\right\rangle $.
This requires us to initialize $r$ with sufficiently small $\left\Vert r\right\Vert _{q}$.
Algorithm \ref{Alg:regularized-regression-1} then follows similarly
to Algorithm \ref{Alg:low-precision}, with the note that it suffices
to obtain only a constant approximation. We give the correctness and
convergence of Algorithm \ref{Alg:regularized-regression-1} in Lemma
\ref{lem:residual-solver-guarantee} whose proof is based on the same
idea as the analysis for Algorithm \ref{Alg:low-precision}. 

The complete analysis of our algorithm is provided in Appendix \ref{sec:Proof-of-Theorem-2}.

\section{Experimental Evaluation}

\begin{figure*}
\subfloat[size=$n\times(n-50),p=8$]{\includegraphics[width=0.25\textwidth]{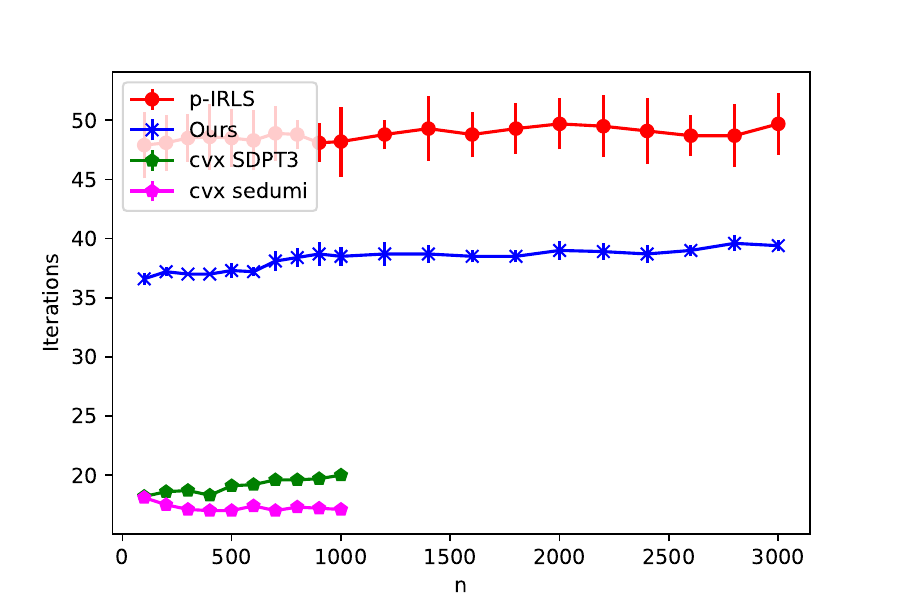}}\subfloat[size=$500\times400$]{\includegraphics[width=0.25\textwidth]{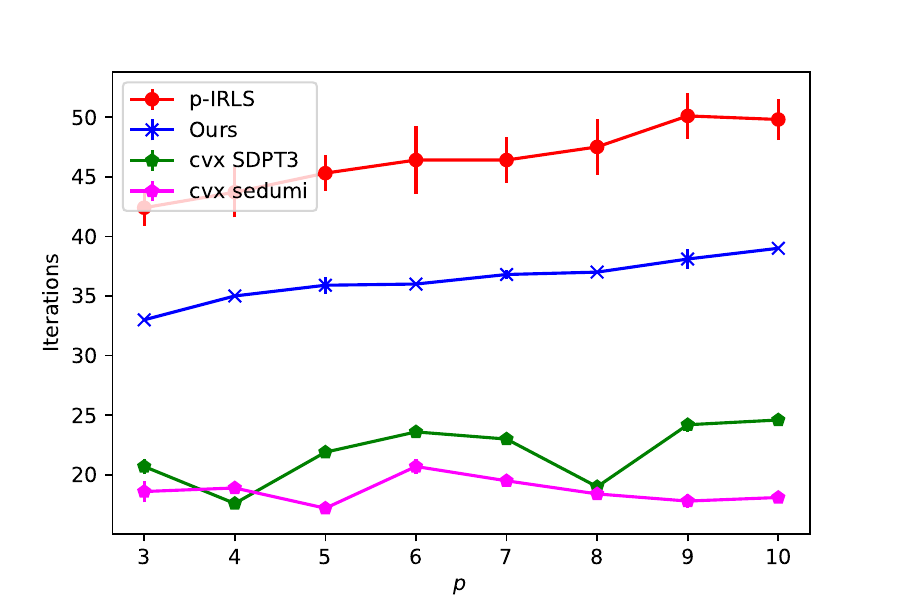}}\subfloat[size=$n\times(n-100)$]{\includegraphics[width=0.25\textwidth]{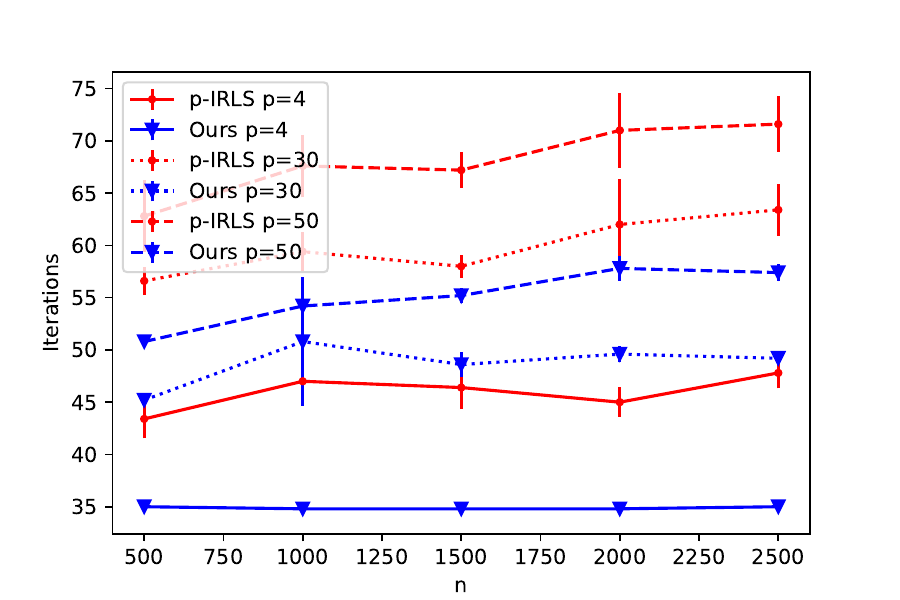}}\subfloat[size=$n\times(n-100)$]{\includegraphics[width=0.25\textwidth]{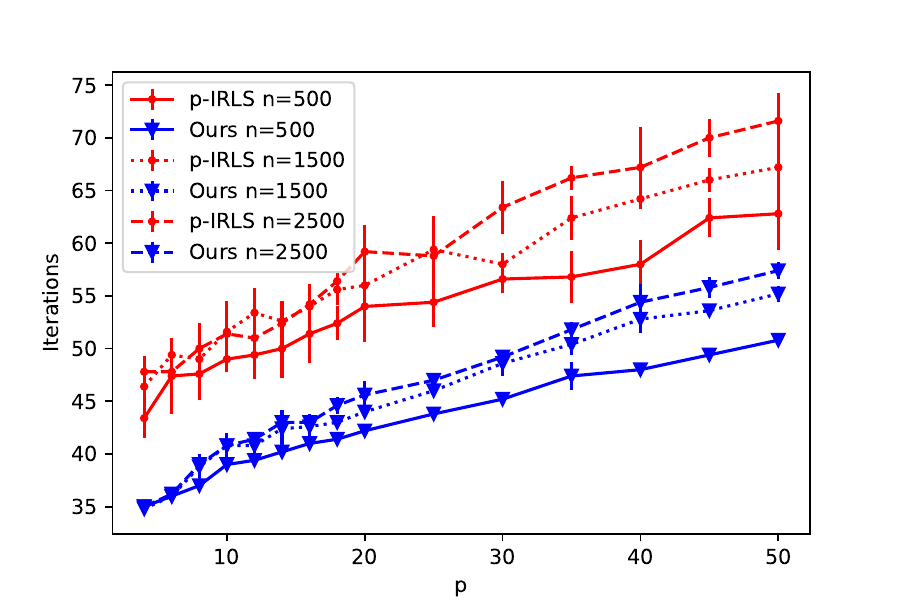}}

\subfloat[size=$n\times(n-50),p=8$]{\includegraphics[width=0.25\textwidth]{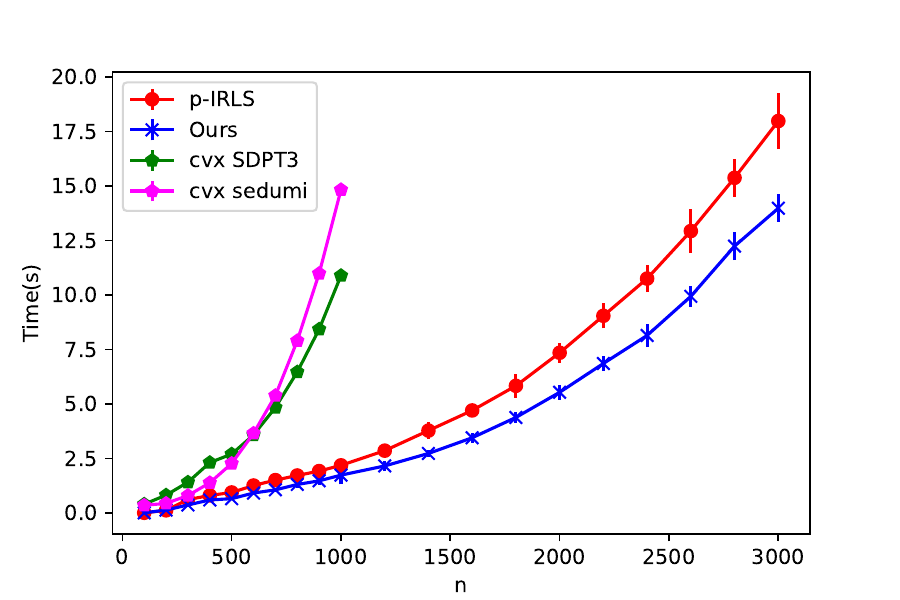}}\subfloat[size=$500\times400$]{\includegraphics[width=0.25\textwidth]{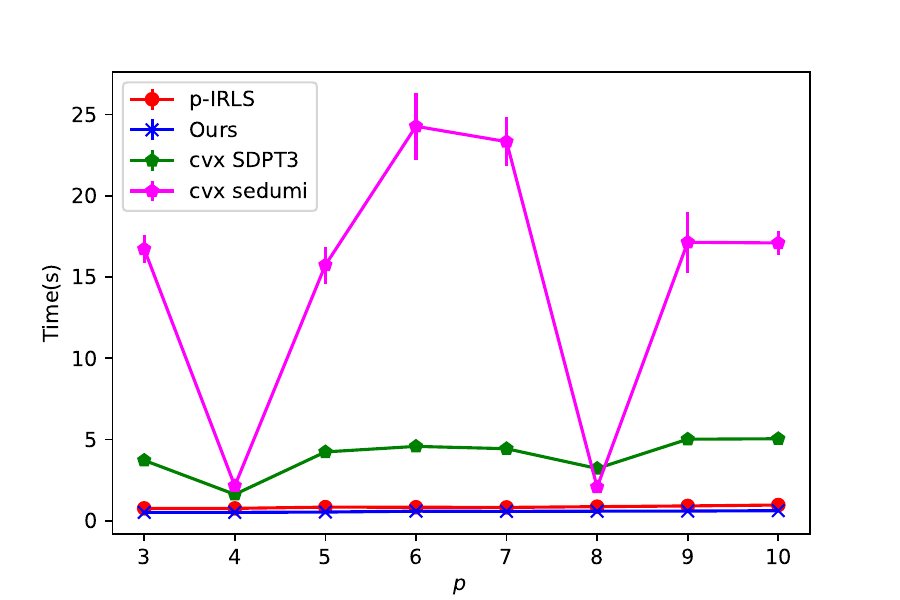}}\subfloat[size=$n\times(n-100)$]{\includegraphics[width=0.25\textwidth]{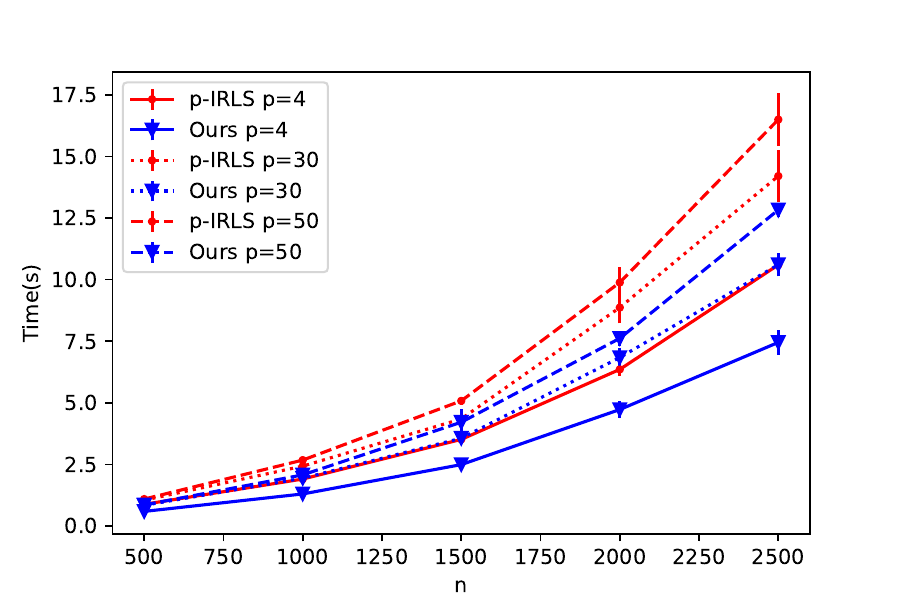}}\subfloat[size=$n\times(n-100)$]{\includegraphics[width=0.25\textwidth]{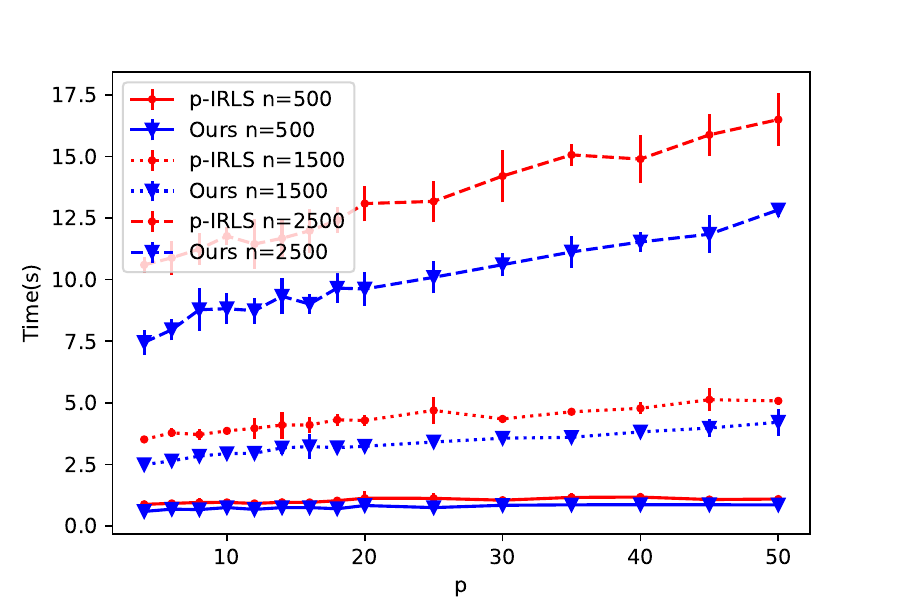}}

\caption{\protect\label{fig:random-matrix}Performance on random matrices:
$\min\left\Vert Ax-b\right\Vert _{p}^{p}$ with $\epsilon=10^{-10}$.
We compare our algorithm with CVX using SDPT3 and SeDuMi solvers and
$p$-IRLS by \citet{adil2019fast}. Figures (a),(b),(e),(f) plot the
average and standard deviation of number of iterations and time taken
by the solvers to find a solution over 10 runs. Figures (c),(d),(g),(h)
measure over 5 runs.}
\end{figure*}

\paragraph*{On synthetic data.}

We follow the experimental setup in \citet{adil2019fast}, and build
on the provided code\footnote{The code is available at \url{https://github.com/fast-algos/pIRLS}}.
We evaluate the performance of our high-precision Algorithm \ref{alg:main-algo}
on the problem $\min\left\Vert Ax-b\right\Vert _{p}^{p}$ on two types
of instances: (1) Random matrices: the entries of $A$ and $b$ are
generated uniformly at randomly between $0$ and $1$, and (2) Random
graphs: We use the procedure in \citet{adil2019fast} to generate
random graphs and the corresponding $A$ and $b$ (the details are
provided in the appendix).

We vary $p$ and the size of the matrices and graphs, while keeping
the error $\epsilon=10^{-10}$. All implementations were done on MATLAB
2024a on a MacBook Pro M2 with 16GB RAM. We measure the number of
iterations and running time for each algorithm and report them in
Figures \ref{fig:random-matrix}-\ref{fig:random-graph}. In the appendix,
we provide additional experimental results when $1<p<2$ and when
$\epsilon$ varies. 

\paragraph*{On real-world datasets.}

We test our algorithm against $p$-IRLS on six regression datasets
 from the UCI repository. CVX has excessive runtime and hence is excluded
from the comparison. Results are provided in Table \ref{tab:real-world}.
\begin{rem}
Regarding the correctness of the algorithm, we use the output by CVX
as the baseline. In all experiments, our algorithm has error within
the $\epsilon$ margin compared with the objective value of the  CVX
solution (see appendix).
\end{rem}
On smaller instances, we compare our algorithm with CVX using SDPT3
and Sedumi solvers and the $p$-IRLS algorithm by \citet{adil2019fast}.
While CVX solvers generally need fewer iterations to find a solution,
they are significantly slower on all instances than our algorithm
and $p$-IRLS. Our algorithm also significantly outperforms $p$-IRLS
in both the number of iterations (calls to a linear system solver)
and running time. When the size of the problem and the value of $p$
increases, the gap between our algorithm and $p$-IRLS also increases.
On average, our algorithm is $1$-$2.6$ times faster than $p$-IRLS.

\begin{figure*}
\subfloat[$p=8$]{\includegraphics[width=0.25\textwidth]{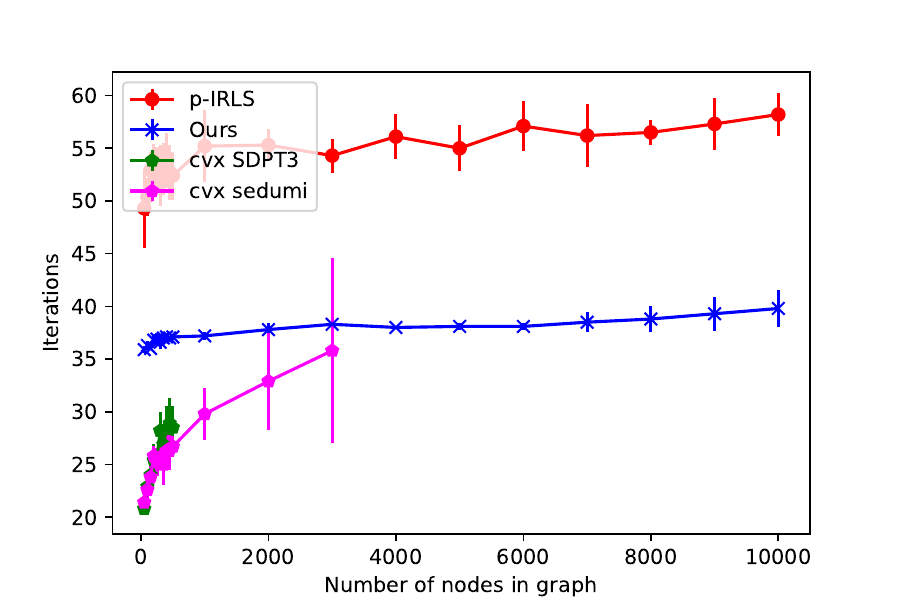}}\subfloat[Number of nodes=$500$]{\includegraphics[width=0.25\textwidth]{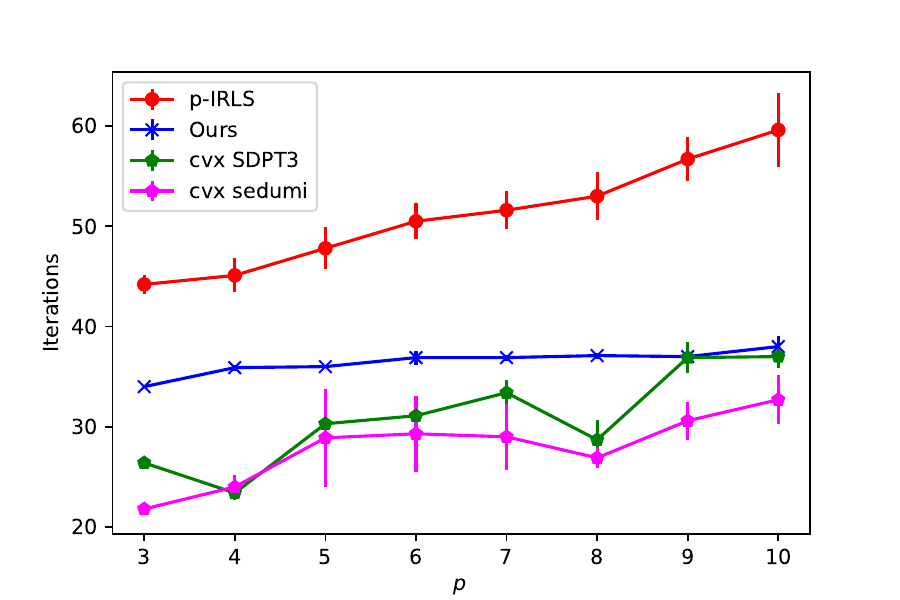}}\subfloat[$n$ nodes]{\includegraphics[width=0.25\textwidth]{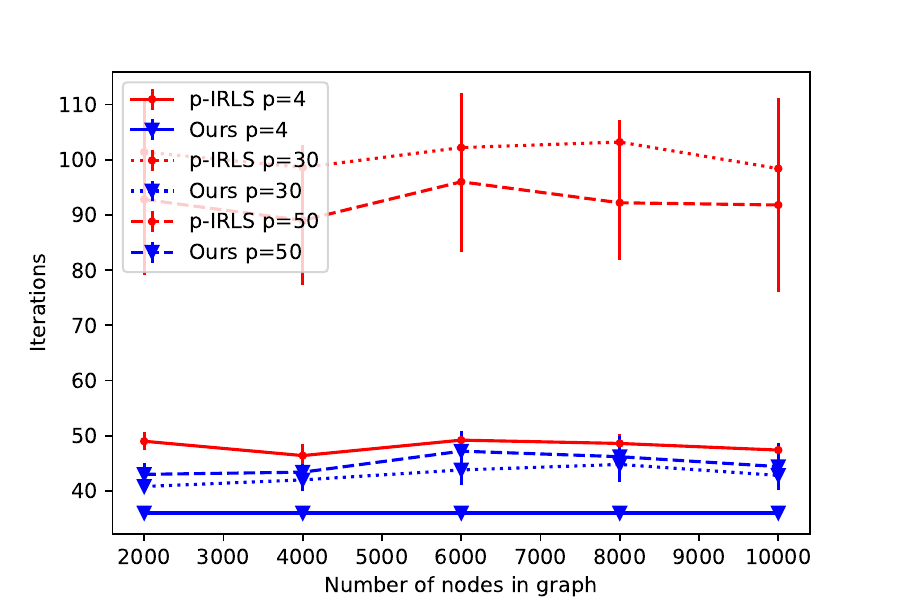}}\subfloat[$n$ nodes]{\includegraphics[width=0.25\textwidth]{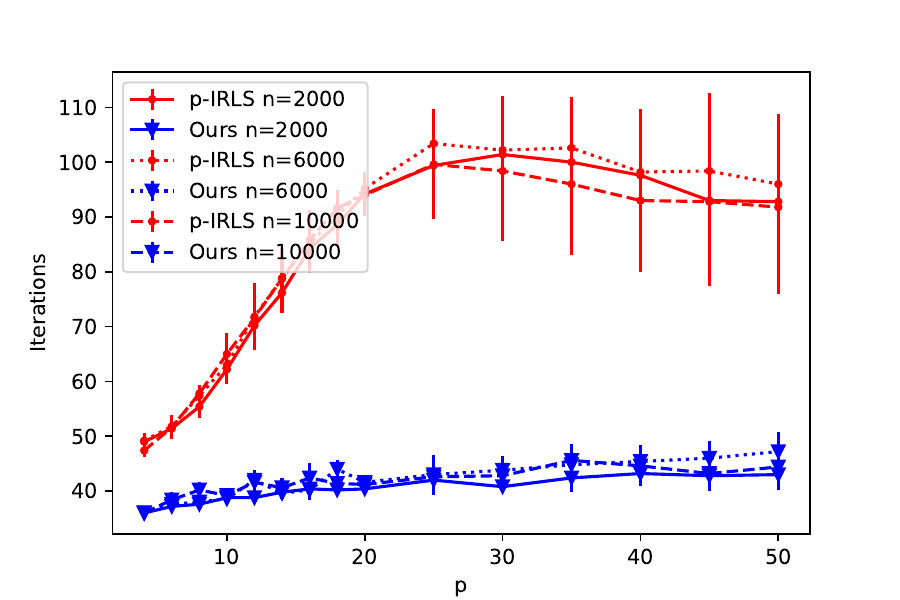}}

\subfloat[$p=8$]{\includegraphics[width=0.25\textwidth]{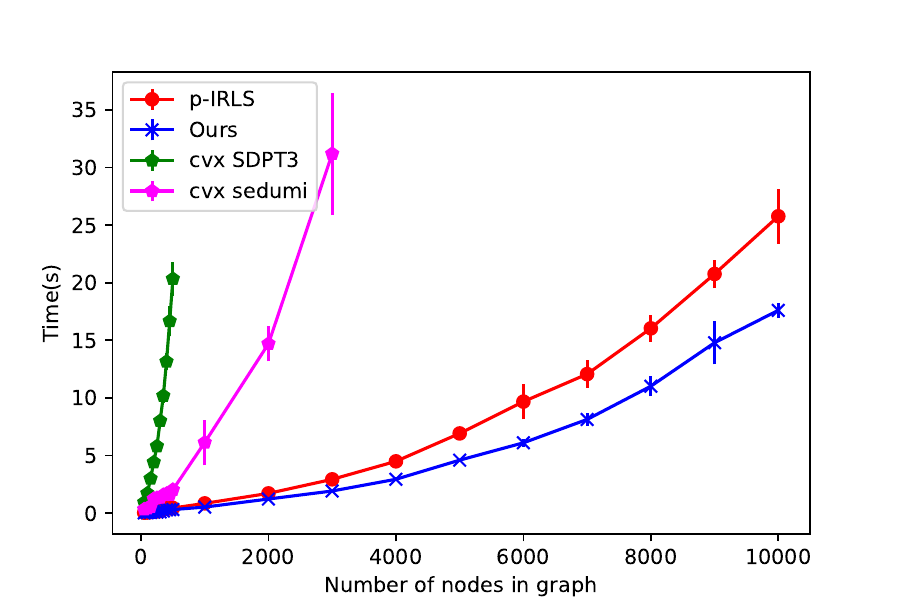}}\subfloat[Number of nodes=$500$]{\includegraphics[width=0.25\textwidth]{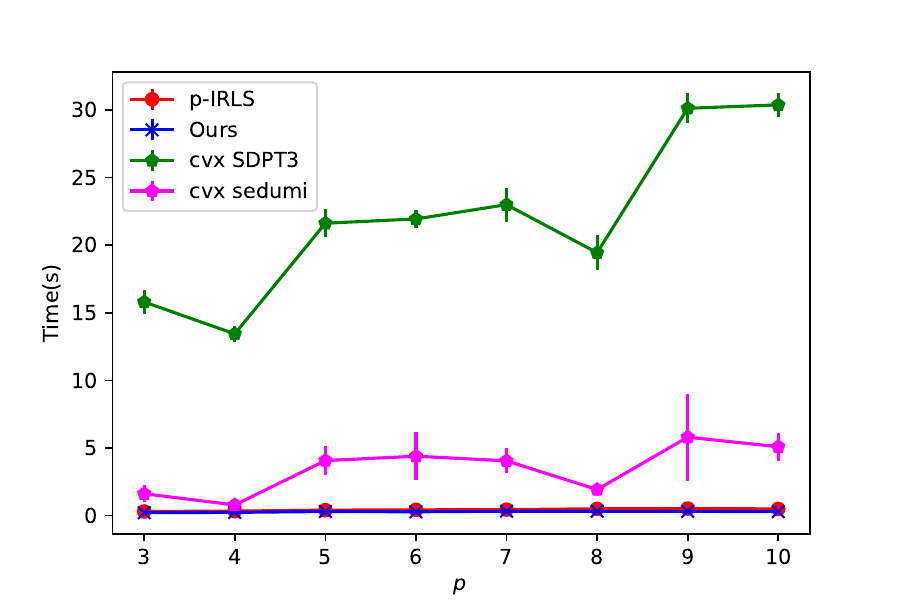}}\subfloat[$n$ nodes]{\includegraphics[width=0.25\textwidth]{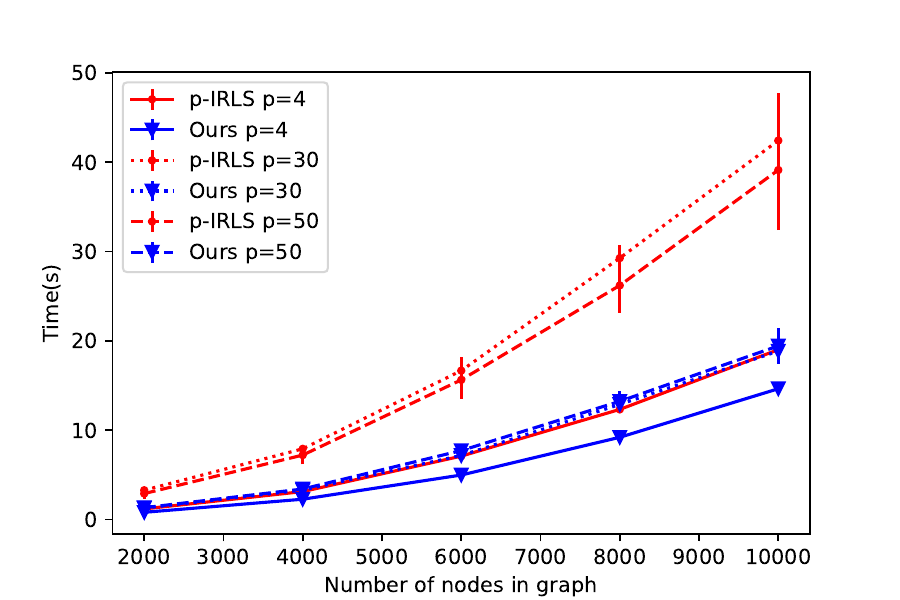}}\subfloat[$n$ nodes]{\includegraphics[width=0.25\textwidth]{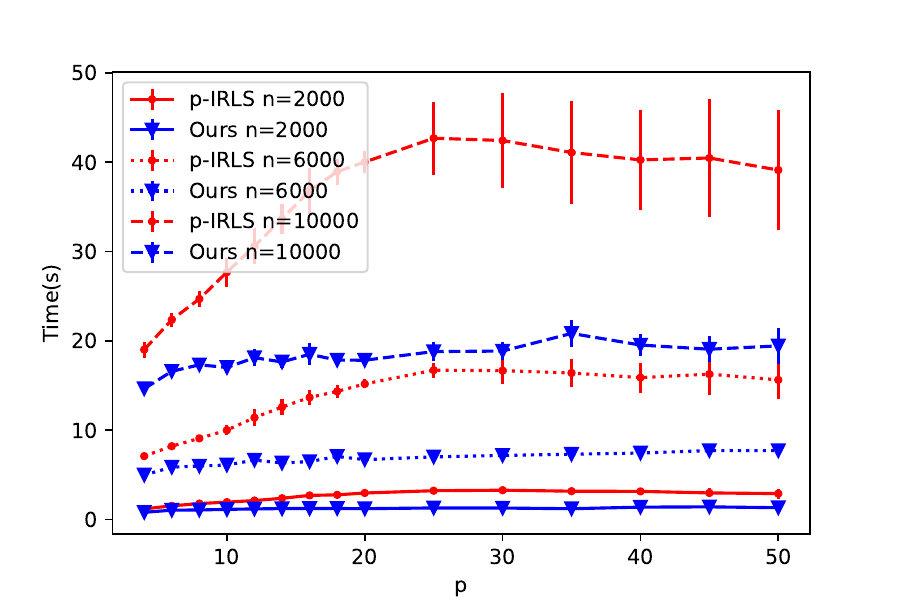}}

\caption{\protect\label{fig:random-graph}Performance on random graph instances:
$\min\left\Vert Ax-b\right\Vert _{p}^{p}$ with $\epsilon=10^{-10}$.
We compare our algorithm with CVX using SDPT3 and SeDuMi solvers and
$p$-IRLS by \citet{adil2019fast}. Figures (a),(b),(e),(f) measure
over 10 runs. Figures (c),(d),(g),(h) measure over 5 runs.}
\end{figure*}

\begin{table}
\caption{Performance of our algorithm against $p$-IRLS on six real-world datasets
for $p=8$, $\epsilon=10^{-10}$. }
\label{tab:real-world}

\begin{tabular*}{1\columnwidth}{@{\extracolsep{\fill}}|>{\centering}m{0.04\columnwidth}|>{\centering}m{0.1\columnwidth}|>{\centering}m{0.1\columnwidth}|>{\centering}m{0.1\textwidth}|>{\centering}m{0.1\textwidth}|>{\centering}m{0.1\columnwidth}|>{\centering}m{0.1\columnwidth}|>{\centering}m{0.1\columnwidth}|}
\hline 
 &  & CT slices \citet{relative_location_of_ct_slices_on_axial_axis_206} & KEGG Metabolic \citet{kegg_metabolic_relation_network_directed_220} & Power Consump-tion \citet{individual_household_electric_power_consumption_235} & Buzz in Social Media \citet{buzz_in_social_media__248} & Protein Property \citet{physicochemical_properties_of_protein_tertiary_structure_265} & Song Year Prediction \citet{year_prediction_msd_203}\tabularnewline
\hline 
\hline 
 & Size & 48150 $\times$385 & 57248 $\times$27 & 1844352 $\times$11 & 524925 $\times$77 & 41157$\times$9 & 463811 $\times$90\tabularnewline
\hline 
\multirow{2}{0.04\columnwidth}{no. iters} & $p$-IRLS & 48 & 50 & 45 & 50 & 44 & 45\tabularnewline
\cline{2-8}
 & Ours & 36 & 42 & 36 & 42 & 36 & 36\tabularnewline
\hline 
\multirow{2}{0.04\columnwidth}{time (s)} & $p$-IRLS & 14.3 & 2.5 & 32. & 28. & 1.6 & 22.5\tabularnewline
\cline{2-8}
 & Ours & 9.2 & 1.7 & 15.7 & 18.1 & 1.1 & 13.3\tabularnewline
\hline 
\end{tabular*}
\end{table}

\bibliographystyle{plainnat}
\bibliography{ref}

\appendix

\section{Property of the Energy Function}

We recall the definition of energy function and its properties used
in the algorithms.
\begin{defn}
(Energy function). Given a vector $r\in\R_{+}^{n}$, we let the electrical
energy be $\mathcal{E}(r)=\min_{x:Ax=b}\langle r,x^{2}\rangle$.
\end{defn}
\begin{lem}
(Computing the energy minimizer) Given $b\in\R^{d}$ and $r\in\R_{+}^{n}$,
the least squares problem $\min_{x:Ax=b}\langle r,x^{2}\rangle$ can
be solved by evaluating $x=\mathbb{D}(r)^{-1}A^{\top}\left(A\mathbb{D}(r)^{-1}A^{\top}\right)^{+}b$,
where $\mathbb{D}(r)$ is the diagonal matrix whose entries are given
by $r$.
\end{lem}
The following lemma gives us a lower bound on the increase in electrical
energy when we increase $r$.
\begin{lem}
\label{lem:energy-inc-basic-1}Given $r'\geq r$ and letting $x=\arg\min_{x:Ax=b}\langle r,x^{2}\rangle$,
one has that
\[
\mathcal{E}(r')-\mathcal{E}(r)\geq\sum_{i}r_{i}x_{i}^{2}\left(1-\frac{r_{i}}{r'_{i}}\right).
\]
\end{lem}
\begin{proof}
This inequality follows from the standard lower bound for $\mathcal{E}(r')-\mathcal{E}(r)$,
which the reader can find in \citet{ene2019improved}. 
\end{proof}

\section{Reducing General Regression Problems to the Affine-Constrained Version
\protect\label{sec:Reducing-General-Regression}}

In this section we show that the affine constrained version of the
problem we consider is in full generality. Formally, we show that
any $\ell_{p}$ regression problem of the form $\min_{Ax=b}\left\Vert Nx-v\right\Vert _{p}$
can be reduced to the form we consider.
\begin{lem}
Let $A\in\mathbb{R}^{s\times n},b\in\mathbb{R}^{s},N\in\mathbb{R}^{m\times n},v\in\mathbb{R}^{m}$
and consider the optimization objective $\min_{Ax=b}\left\Vert Nx-v\right\Vert _{p}$.
Let $\left[\begin{array}{c}
x\\
z
\end{array}\right]$ be a $(1+\varepsilon)$ approximate solution to the affine-constrained
regression problem
\[
\min_{\left[\begin{array}{cc}
N & -I_{m\times m}\\
A & 0_{s\times m}
\end{array}\right]\left[\begin{array}{c}
x\\
z
\end{array}\right]=\left[\begin{array}{c}
v\\
b
\end{array}\right]}\left\Vert z\right\Vert _{p}\,.
\]
Then $x$ is a $(1+\varepsilon)$ approximate solution to the original
objective. Furthermore, each least squares subproblem can be solved
using two calls to a linear system solver for $N^{\top}RN$, and one
call to a linear system solver for $A\left(N^{\top}RN\right)^{+}A^{\top}$.
\end{lem}
\begin{proof}
We augment the dimension of the iterate by introducing $m$ additional
variables encoded in a vector $z\in\mathbb{R}^{m}$. Hence one can
equivalently enforce the constraints
\begin{align*}
Nx-z & =v\\
Ax & =b
\end{align*}
and simply seek to minimize $\left\Vert z\right\Vert _{p}$ instead
of $\left\Vert Ax-b\right\Vert _{p}$, which is the suitable formulation
required by our solver. Note that while we do not have any weights
on the $x$ iterate, the analysis goes through normally, since in
fact it tolerates solving a more general weighted $\ell_{p}$ regression
problem. 

To solve the corresponding least squares problem, we need to compute
\begin{align*}
\min_{Ax=b}\frac{1}{2}\left\langle r,\left(Nx-v\right)^{2}\right\rangle  & =\min_{Ax=b}\frac{1}{2}x^{\top}N^{\top}RNx-\left\langle N^{\top}Rv,x\right\rangle +\frac{1}{2}v^{\top}Rv\\
 & =\max_{y}\min_{x}\frac{1}{2}x^{\top}N^{\top}RNx-\left\langle N^{\top}Rv,x\right\rangle +\frac{1}{2}v^{\top}Rv+\left\langle b-Ax,y\right\rangle \\
 & =\max_{y}\left(\left\langle b,y\right\rangle +\min_{x}\frac{1}{2}x^{\top}N^{\top}RNx-\left\langle N^{\top}Rv+A^{\top}y,x\right\rangle \right)-\frac{1}{2}v^{\top}Rv\,.
\end{align*}
where $R$ is the diagonal matrix whose entries are given by $r$.
The inner problem is minimized at 
\[
x=\left(N^{\top}RN\right)^{+}\left(N^{\top}Rv+A^{\top}y\right)\,,
\]
which simplifies the problem to
\begin{align*}
\max_{y} & \left\langle b,y\right\rangle -\frac{1}{2}\left(N^{\top}Rv+A^{\top}y\right)^{\top}\left(N^{\top}RN\right)^{+}\left(N^{\top}Rv+A^{\top}y\right)-\frac{1}{2}v^{\top}Rv\\
=\max_{y} & \left\langle b-A\left(N^{\top}RN\right)^{+}N^{\top}Rv,y\right\rangle -\frac{1}{2}y^{\top}A\left(N^{\top}RN\right)^{+}A^{\top}y\\
 & -\frac{1}{2}v^{\top}RN\left(N^{\top}RN\right)^{+}N^{\top}Rv-\frac{1}{2}v^{\top}Rv\,,
\end{align*}
which is maximized at
\[
y=\left(A\left(N^{\top}RN\right)^{+}A^{\top}\right)^{+}\left(b-A\left(N^{\top}RN\right)^{+}N^{\top}Rv\right)\,,
\]
so
\begin{align*}
x & =\left(N^{\top}RN\right)^{+}N^{\top}Rv+\left(N^{\top}RN\right)^{+}A^{\top}\left(A\left(N^{\top}RN\right)^{+}A^{\top}\right)^{+}\left(b-N\left(N^{\top}RN\right)^{+}N^{\top}Rv\right)\\
 & =\left(N^{\top}RN\right)^{+}\left(N^{\top}Rv+A^{\top}\left(A\left(N^{\top}RN\right)^{+}A^{\top}\right)^{+}\left(b-A\left(N^{\top}RN\right)^{+}N^{\top}Rv\right)\right)\,.
\end{align*}
We observer that to execute this step we require two calls to a solver
for $N^{\top}RN$, and one call to a solver for $A\left(N^{\top}RN\right)^{+}A^{\top}$.
\end{proof}

\section{Solving $\ell_{p}$ Regression for $1\protect\leq p<2$ \protect\label{sec:small-p-reduction}}

In this section we show that while our solvers are defined for $\ell_{p}$
regression when $p\geq2$, they also provide solutions $\ell_{q}$
regression for $1\leq q<2$. This follows directly from exploiting
duality. See \citet{adil2019iterative}, section 7.2 for a proof detailed.
Here we briefly explain why this is the case. Let $p,q$ such that
$\frac{1}{p}+\frac{1}{q}=1$, $1\leq q<2$, and consider the $\ell_{q}$
regression problem, along with its dual
\[
\min_{x:Ax=b}\left\Vert x\right\Vert _{q}=\max_{\left\Vert A^{\top}y\right\Vert _{p}\leq1}\left\langle b,y\right\rangle .
\]
We can use our solver to provide a high precision solution to the
dual maximization problem, which we then show can be used to read
off a primal nearly optimal solution. Indeed, we can equivalently
solve
\[
\min_{\left\langle b,y\right\rangle =1}\left\Vert A^{\top}y\right\Vert _{p}
\]
to high precision $\varepsilon=\frac{1}{n^{O\left(1\right)}}$ , based
on which we construct the nearly-feasible primal solution
\[
x=\frac{\left\langle b,y\right\rangle }{\left\Vert A^{\top}y\right\Vert _{p}^{p}}\cdot\left(A^{\top}y\right)^{p-1}\,.
\]
To see why this is a good solution, let us assume that we achieve
exact gradient optimality for $y$, which means that for some scalar
$\lambda$,
\begin{equation}
A\left(A^{\top}y\right)^{p-1}=b\cdot\lambda\,.\label{eq:grad-opt-dual}
\end{equation}
First let us verify that $x$ is feasible. Using $(\ref{eq:grad-opt-dual})$
we see that:
\[
Ax=A\left(\frac{\left\langle b,y\right\rangle }{\left\Vert A^{\top}y\right\Vert _{p}^{p}}\cdot\left(A^{\top}y\right)^{p-1}\right)=\frac{\left\langle b,y\right\rangle }{\left\Vert A^{\top}y\right\Vert _{p}^{p}}\cdot A\left(A^{\top}y\right)^{p-1}=\left(\frac{\left\langle b,y\right\rangle }{\left\Vert A^{\top}y\right\Vert _{p}^{p}}\cdot\lambda\right)\cdot b\,.
\]
Additionally we can also use $(\ref{eq:grad-opt-dual})$ again to
obtain that
\[
\left\Vert A^{\top}y\right\Vert _{p}^{p}=\left\langle y,A\left(A^{\top}y\right)^{p-1}\right\rangle =\left\langle y,b\right\rangle \cdot\lambda\,,
\]
which allows us to conclude that 
\[
Ax=b\,,
\]
so $x$ is feasible. Finally, we can measure the duality gap by calculating
\begin{align*}
\left\Vert x\right\Vert _{q} & =\frac{1}{\lambda}\left\Vert \left(A^{\top}y\right)^{p-1}\right\Vert _{q}=\frac{1}{\lambda}\cdot\left(\sum\left(A^{\top}y\right)^{\left(p-1\right)\frac{p}{p-1}}\right)^{\frac{p-1}{p}}=\frac{1}{\lambda}\left\Vert A^{\top}y\right\Vert _{p}^{p-1}\\
 & =\frac{\left\langle y,b\right\rangle }{\left\Vert A^{\top}y\right\Vert _{p}^{p}}\cdot\left\Vert A^{\top}y\right\Vert _{p}^{p-1}=\frac{\left\langle y,b\right\rangle }{\left\Vert A^{\top}y\right\Vert _{p}}\,,
\end{align*}
which certifies optimality for $b$. While in general we do not solve
the dual problem exactly, which yields a slight violation in the demand
for the primal iterate $x$, this can be fixed by adding to $x$ a
correction $\widetilde{x}=A^{\top}\left(AA^{\top}\right)^{+}\left(b-Ax\right)$
that handles the residual $b-Ax$. This affects the $\ell_{q}$ norm
only slightly since the residual is guaranteed to be very small due
to the near-optimality of the dual problem. Then we can proceed to
bounding the duality gap by following the argument sketched above,
while also carrying the polynomially small error through the calculation.
We refer the reader to \citet{adil2019iterative} for the detailed
error analysis.

\section{Proof of Theorem \ref{thm:low-precision} \protect\label{sec:Proof-of-Theorem-1}}

In this section, we first outline the necessary lemmas needed to prove
Theorem \ref{thm:low-precision} before providing their proofs below.

\paragraph{Correctness of Algorithm \ref{Alg:low-precision}. }

There are two possible outcomes of Algorithm \ref{Alg:low-precision}.
Either it returns a primal solution (Case 1 and Case 2) or a dual
certificate (Case 3). In the former two cases, Case 2 immediately
gives us an approximate solution. We show in Lemma \ref{lem:valid-rule}
that the returned vector in Case 1 achieves the target approximation
guarantee. In Case 3, we use the invariant shown in Lemma \ref{lem:invariant}
to show that the returned dual solution is an infeasibility certificate.

We formalize these statements in the lemmas below.
\begin{lem}[Invariant]
\label{lem:invariant}For all $t$, we have that if $\gamma^{(t)}\neq1$
then $\frac{\mathcal{E}(r^{(t+1)})-\mathcal{E}(r^{(t)})}{\left\Vert r^{(t+1)}\right\Vert _{q}-\left\Vert r^{(t)}\right\Vert _{q}}\ge M^{2}$.
\end{lem}
\begin{lem}[Case 1]
\label{lem:valid-rule} Let $r$ be a dual solution and $x=\arg\min_{\widehat{x}:A\widehat{x}=b}\langle r,\widehat{x}^{2}\rangle$.
If $\left\Vert \left\Vert r\right\Vert _{q}^{q-1}\cdot\frac{x^{2}}{r^{q-1}}\right\Vert _{\infty}\leq\left(1+\epsilon\right)M^{2}$
then $\left\Vert x\right\Vert _{2p}\leq M(1+\epsilon)$.
\end{lem}
\begin{lem}[Case 3]
\label{lem:resistance-increase-average}If the algorithm returns
$r^{(T)}$, then $\frac{\mathcal{E}\left(r^{(T)}\right)}{\left\Vert r^{(T)}\right\Vert _{q}}\geq\frac{M^{2}}{(1+\epsilon)^{2}}.$
\end{lem}

\paragraph{Convergence of Algorithm \ref{Alg:low-precision}.}

We run the algorithm for $T$ iterations. The algorithm terminates
if at any point it finds a solution $x$ that satisfies the desired
bound (otherwise it is unable to further increase the dual solution).
Otherwise, we show that it must finish very fast. Suppose we run it
for $T=T_{hi}+T_{lo}$ iterations. Let the iterations in $T_{hi}$
correspond to those where at least a single coordinate of $r$ was
scaled by $\geq S\coloneqq n^{\frac{2}{2q+1}}\left(\frac{1}{\epsilon}\right)^{\frac{q-1}{2q+1}}$.
Let $T_{lo}$ be the remaining iterations. The following lemmas give
an upperbound on $T_{hi}$ and $T_{lo}$.
\begin{lem}
\label{lem:bound-t-hi}We have $T_{hi}\leq\frac{n}{S^{q}\epsilon^{q}}$.
\end{lem}
\begin{lem}
\label{lem:bound-t-lo}We have $T_{lo}\leq O\left(\left(\frac{1}{\epsilon}+\frac{S^{1/2}}{q\ln S}\right)\frac{1}{\epsilon^{\frac{q+1}{2}}}\log\left(\frac{n}{\epsilon^{q}}\right)\right)$.
\end{lem}
Since $S=n^{\frac{2}{2q+1}}\left(\frac{1}{\epsilon}\right)^{\frac{q-1}{2q+1}}$,
we obtain the following convergence guarantee:
\begin{lem}
\label{lem:final-time}Algorithm \ref{Alg:low-precision} terminates
in $O\left(\left(\left(\frac{1}{\epsilon}\right)^{\frac{q+3}{2}}+n^{\frac{1}{2q+1}}\left(\frac{1}{\epsilon}\right)^{\frac{q^{2}+2q}{2q+1}}\right)\log\left(\frac{n}{\epsilon^{q}}\right)\right)$
iterations.
\end{lem}
Equipped with these lemmas, we give the proof for Theorem \ref{thm:low-precision}.
\begin{proof}[Proof of Theorem \ref{thm:low-precision}]
Returning to the problem $\min_{x\in\R^{n}\colon Ax=b}\|x\|_{p}$,
we have the main algorithm executes a binary search over the power
of $(1+\epsilon)$ in the range $\left[\frac{\left\Vert x^{(0)}\right\Vert _{2}}{n^{\frac{1}{2}-\frac{1}{p}}},\left\Vert x^{(0)}\right\Vert _{2}\right]$,
so the total number of calls to the subroutine solver is $O\left(\log\log n+\log\frac{1}{\epsilon}\right)$.
By Lemma \ref{lem:final-time}, the subroutine solver requires $O\left(\left(\left(\frac{1}{\epsilon}\right)^{\frac{q+3}{2}}+n^{\frac{1}{2q+1}}\left(\frac{1}{\epsilon}\right)^{\frac{q^{2}+2q}{2q+1}}\right)\log\left(\frac{n}{\epsilon^{q}}\right)\right)$
linear system solves, where $q=\frac{p}{p-2}$ is the dual norm of
$p/2$. Substituting the value of $q$, we obtain the conclusion.
\end{proof}

\subsection{Proofs of Lemmas \ref{lem:invariant} - \ref{lem:bound-t-lo}}

\begin{proof}[Proof of Lemma \ref{lem:invariant}]
First we show (\ref{eq:resistance-increase}). 
\begin{align*}
\frac{1}{\left\Vert r^{(t+1)}\right\Vert _{q}-\left\Vert r^{(t)}\right\Vert _{q}} & \ge\frac{q\left\Vert r^{(t)}\right\Vert _{q}^{q-1}}{\left\Vert r^{(t+1)}\right\Vert _{q}^{q}-\left\Vert r^{(t)}\right\Vert _{q}^{q}}.
\end{align*}
This is equivalent to show 
\begin{align*}
\left\Vert r^{(t+1)}\right\Vert _{q}^{q}+\left(q-1\right)\left\Vert r^{(t)}\right\Vert _{q}^{q} & \ge q\left\Vert r^{(t+1)}\right\Vert _{q}\left\Vert r^{(t)}\right\Vert _{q}^{q-1}
\end{align*}
which can easily be obtained from AM-GM inequality.

Using (\ref{eq:resistance-increase}) and Lemma \ref{lem:energy-inc-basic-1}
we have
\begin{align*}
\frac{\mathcal{E}(r^{(t+1)})-\mathcal{E}(r^{(t)})}{\left\Vert r^{(t+1)}\right\Vert _{q}-\left\Vert r^{(t)}\right\Vert _{q}} & \ge\frac{q\left\Vert r^{(t)}\right\Vert _{q}^{q-1}\left(\sum_{i}r_{i}^{(t)}\left(x_{i}^{(t)}\right)^{2}\left(1-\frac{r_{i}^{(t)}}{r_{i}^{(t+1)}}\right)\right)}{\sum_{i}\left(r_{i}^{(t+1)}\right)^{q}-\left(r_{i}^{(t)}\right)^{q}}\\
 & =\frac{q\left\Vert r^{(t)}\right\Vert _{q}^{q-1}\left(\sum_{i,\alpha_{i}^{(t)}>1}r_{i}^{(t)}\left(x_{i}^{(t)}\right)^{2}\left(1-\frac{r_{i}^{(t)}}{r_{i}^{(t+1)}}\right)\right)}{\sum_{i,\alpha_{i}^{(t)}>1}\left(r_{i}^{(t+1)}\right)^{q}-\left(r_{i}^{(t)}\right)^{q}}.
\end{align*}
For $i$ such that $\alpha_{i}^{(t)}>1$, we have $r_{i}^{(t+1)}=\alpha_{i}^{(t)}r_{i}^{(t)}$,
thus
\begin{align*}
\frac{q\left\Vert r^{(t)}\right\Vert _{q}^{q-1}r_{i}^{(t)}\left(x_{i}^{(t)}\right)^{2}\left(1-\frac{r_{i}^{(t)}}{r_{i}^{(t+1)}}\right)}{\left(r_{i}^{(t+1)}\right)^{q}-\left(r_{i}^{(t)}\right)^{q}} & =\frac{\left\Vert r^{(t)}\right\Vert _{q}^{q-1}\left(x_{i}^{(t)}\right)^{2}}{\left(r_{i}^{(t)}\right)^{q-1}}\cdot\frac{q\left(1-\frac{1}{\alpha_{i}^{(t)}}\right)}{\left(\alpha_{i}^{(t)}\right)^{q}-1}\\
 & \ge\gamma_{i}^{(t)}M^{2}\cdot\frac{1}{\left(\alpha_{i}^{(t)}\right)^{q}}\\
 & =M^{2},
\end{align*}
where the first inequality is due to $\frac{q\left(\alpha-1\right)}{\alpha\left(\alpha^{q}-1\right)}\ge\frac{1}{\alpha^{q}}$,
for $\alpha>1$. We can then obtain the desired conclusion from here.
\end{proof}

\begin{proof}[Proof of Lemma \ref{lem:valid-rule}]
If
\begin{align*}
\left\Vert \left\Vert r\right\Vert _{q}^{q-1}\cdot\frac{x^{2}}{r^{q-1}}\right\Vert _{\infty} & \leq\left(1+\epsilon\right)M^{2},
\end{align*}
for all $i$ we have 
\begin{align*}
x_{i}^{2} & \le\left(1+\epsilon\right)^{2}M^{2}\frac{r_{i}^{q-1}}{\left\Vert r\right\Vert _{q}^{q-1}},
\end{align*}
which gives 
\begin{align*}
x_{i}^{2p} & \le\left(1+\epsilon\right)^{2p}M^{2p}\frac{r_{i}^{q}}{\left\Vert r\right\Vert _{q}^{q}},
\end{align*}
We obtain 
\begin{align*}
\left\Vert x\right\Vert _{2p}^{2p} & \le\left(1+\epsilon\right)^{2p}M^{2p},
\end{align*}
as needed.
\end{proof}

\begin{proof}[Proof of Lemma \ref{lem:resistance-increase-average}]
We have that
\begin{align*}
\frac{\mathcal{E}(r^{(T)})}{\left\Vert r^{(T)}\right\Vert _{q}} & =\frac{\mathcal{E}(r^{(0)})+\sum_{t=0}^{T-1}\left(\mathcal{E}(r^{(t+1)})-\mathcal{E}(r^{(t)})\right)}{\left\Vert r^{(T)}\right\Vert _{q}}\\
 & \geq\frac{\mathcal{E}(r^{(0)})+\sum_{t=0}^{T-1}\left(\left\Vert r^{(t+1)}\right\Vert _{q}-\left\Vert r^{(t)}\right\Vert _{q}\right)\cdot M^{2}}{\left\Vert r^{(T)}\right\Vert _{q}}\quad\text{(due to the invariant)}\\
 & \geq\frac{\left(\left\Vert r^{(T)}\right\Vert _{q}-1\right)\cdot M^{2}}{\left\Vert r^{(T)}\right\Vert _{q}}=M^{2}\cdot\left(1-\frac{1}{\left\Vert r^{(T)}\right\Vert _{q}}\right)\\
 & \geq M^{2}\cdot\left(1-\epsilon\right)\quad\hfill\text{(since \ensuremath{\left\Vert r^{(T)}\right\Vert _{q}\geq\frac{1}{\epsilon}})}\\
 & \ge\frac{M^{2}}{\left(1+\epsilon\right)^{2}}.
\end{align*}
\end{proof}

\begin{proof}[Proof of Lemma \ref{lem:bound-t-hi}]
Suppose the contrary. Then we claim that the perturbations that scale
the dual solution by $\ge S$ will have increased it a lot to the
point where $\left\Vert r\right\Vert _{q}^{q}\geq\frac{1}{\epsilon^{q}}$.
Indeed, since $r$ is initialized to $\frac{1}{n^{1/q}}$, in the
worst case each perturbation in $T_{hi}$ touches a different coordinate
$i$. Therefore this establishes a lower bound of $T_{hi}\cdot\frac{S^{q}}{n}$
on $\left\Vert r\right\Vert _{q}^{q}$. As this must be at most $\frac{1}{\epsilon^{q}}$,
since otherwise we obtained a good solution per Lemma \ref{lem:resistance-increase-average},
we obtain the conclusion.
\end{proof}

Before showing the proof of Lemma \ref{lem:bound-t-lo}, we claim
that we can either look at the history produced in $T_{lo}$ and obtain
an approximately feasible solution, or a single coordinate of $r$
must have increased a lot.
\begin{lem}
\label{lem-t-lo-cert}Consider the set of iterates $(r^{(t)},x^{(t)})$
used for the iterates in $T_{lo}$. If 
\[
\left\Vert \frac{1}{T_{lo}}\sum_{t\in T_{lo}}x^{(t)}\right\Vert _{2p}>M(1+\epsilon)
\]
then there exists a coordinate $i$ for which 
\[
\sum_{t\in T_{lo}:\alpha_{i}^{(t)}>1}\sqrt{\alpha_{i}^{(t)}}\geq\frac{T_{lo}\epsilon^{\frac{q+1}{2}}}{2}.
\]
\end{lem}
\begin{proof}
Suppose that 
\begin{align*}
\left\Vert \frac{1}{T_{lo}}\sum_{t\in T_{lo}}x^{(t)}\right\Vert _{2p} & >M(1+\epsilon)
\end{align*}
Note that by the update rule,
\begin{align*}
\frac{x_{i}^{(t)}}{M} & \le\left(1+\epsilon\right)^{\frac{1}{2}}\sqrt{\frac{\left(r_{i}^{(t)}\right)^{q-1}}{\left\Vert r^{(t)}\right\Vert _{q}^{q-1}}}+\boldsymbol{1}_{\alpha_{i}>1}\sqrt{\frac{\alpha_{i}^{(t)q}\left(r_{i}^{(t)}\right)^{q-1}}{\left\Vert r^{(t)}\right\Vert _{q}^{q-1}}}\\
 & \le\left(1+\frac{\epsilon}{2}\right)\sqrt{\frac{\left(r_{i}^{(t)}\right)^{q-1}}{\left\Vert r^{(t)}\right\Vert _{q}^{q-1}}}+\boldsymbol{1}_{\alpha_{i}>1}\sqrt{\frac{\alpha_{i}^{(t)q}\left(r_{i}^{(t)}\right)^{q-1}}{\left\Vert r^{(t)}\right\Vert _{q}^{q-1}}}
\end{align*}
Hence we can write 
\begin{align*}
\left\Vert \sum_{t\in T_{lo}}\frac{x^{(t)}}{M}\right\Vert _{2p} & \le\left\Vert \left(1+\frac{\epsilon}{2}\right)\sum_{t\in T_{lo}}\sqrt{\frac{\left(r^{(t)}\right)^{q-1}}{\left\Vert r^{(t)}\right\Vert _{q}^{q-1}}}+\overrightarrow{\left(\sum_{t\in T_{lo},\alpha_{i}^{(t)}>1}\sqrt{\frac{\alpha_{i}^{(t)q}\left(r_{i}^{(t)}\right)^{q-1}}{\left\Vert r^{(t)}\right\Vert _{q}^{q-1}}}\right)_{i}}\right\Vert _{2p}\\
 & \le\left(1+\frac{\epsilon}{2}\right)\sum_{t\in T_{lo}}\left\Vert \sqrt{\frac{\left(r^{(t)}\right)^{q-1}}{\left\Vert r^{(t)}\right\Vert _{q}^{q-1}}}\right\Vert _{2p}+\left\Vert \overrightarrow{\left(\sum_{t\in T_{lo},\alpha_{i}^{(t)}>1}\sqrt{\frac{\alpha_{i}^{(t)q}\left(r_{i}^{(t)}\right)^{q-1}}{\left\Vert r^{(t)}\right\Vert _{q}^{q-1}}}\right)_{i}}\right\Vert _{2p}\\
 & \text{ \ensuremath{\hfill}(by triangle inequality)}\\
 & =\left(1+\frac{\epsilon}{2}\right)T_{lo}+\left\Vert \overrightarrow{\left(\sum_{t\in T_{lo},\alpha_{i}^{(t)}>1}\sqrt{\frac{\alpha_{i}^{(t)q}\left(r_{i}^{(t)}\right)^{q-1}}{\left\Vert r^{(t)}\right\Vert _{q}^{q-1}}}\right)_{i}}\right\Vert _{2p}.
\end{align*}
We obtain 
\begin{align*}
\left\Vert \overrightarrow{\left(\sum_{t\in T_{lo},\alpha_{i}^{(t)}>1}\sqrt{\frac{\alpha_{i}^{(t)q}\left(r_{i}^{(t)}\right)^{q-1}}{\left\Vert r^{(t)}\right\Vert _{q}^{q-1}}}\right)_{i}}\right\Vert _{2p} & \ge\frac{\epsilon}{2}T_{lo}
\end{align*}
On the other hand, we have 
\begin{align*}
\sum_{i}\left(\sum_{t\in T_{lo},\alpha_{i}^{(t)}>1}\sqrt{\frac{\alpha_{i}^{(t)q}\left(r_{i}^{(t)}\right)^{q-1}}{\left\Vert r^{(t)}\right\Vert _{q}^{q-1}}}\right)^{2p} & =\sum_{i}\left(\sum_{t\in T_{lo},\alpha_{i}^{(t)}>1}\sqrt{\frac{\alpha_{i}^{(t)}\left(r_{i}^{(t+1)}\right)^{q-1}}{\left\Vert r^{(t)}\right\Vert _{q}^{q-1}}}\right)^{2p}\\
\le\sum_{i}\left(r_{i}^{(T)}\right)^{q}\left(\sum_{t\in T_{lo},\alpha_{i}^{(t)}>1}\sqrt{\alpha_{i}^{(t)}}\right)^{2p} & \le\left\Vert r^{(T)}\right\Vert _{q}^{q}\max_{i}\left(\sum_{t\in T_{lo},\alpha_{i}^{(t)}>1}\sqrt{\alpha_{i}^{(t)}}\right)^{2p}\\
 & \le\frac{1}{\epsilon^{q}}\max_{i}\left(\sum_{t\in T_{lo},\alpha_{i}^{(t)}>1}\sqrt{\alpha_{i}^{(t)}}\right)^{2p}
\end{align*}
\end{proof}
Therefore there exists $i$ such that 
\begin{align*}
\left(\sum_{t\in T_{lo},\alpha_{i}^{(t)}>1}\sqrt{\alpha_{i}^{(t)}}\right)^{2p} & \ge\left(\frac{\epsilon T}{2}\right)^{2p}\epsilon^{q},
\end{align*}
which gives us
\begin{align*}
\sum_{t\in T_{lo},\alpha_{i}^{(t)}>1}\sqrt{\alpha_{i}^{(t)}} & \ge\frac{T_{lo}\epsilon^{\frac{q+1}{2}}}{2}.
\end{align*}

Now we show the proof of Lemma \ref{lem:bound-t-lo}.

\begin{proof}[Proof of Lemma \ref{lem:bound-t-lo}]
From Lemma \ref{lem-t-lo-cert} we know that there exists a coordinate
$i$ for which 
\begin{equation}
\sum_{t\in T_{lo}:\alpha_{i}^{(t)}>1}\sqrt{\alpha_{i}^{(t)}}>\frac{T_{lo}\epsilon^{\frac{q+1}{2}}}{2}.\label{eq:sum-lb}
\end{equation}
Furthermore by definition for all iterates in $T_{lo}$ we have that
pointwise $\left(1+\epsilon\right)\leq\left(\alpha_{i}^{(t)}\right)^{q}\leq S^{q}$.
This enables us to lower bound the final value of $\left(r_{i}^{(T)}\right)^{q}$
which is a lower bound on $\left\Vert r^{(T)}\right\Vert _{q}^{q}$.
More precisely, we have
\begin{align}
\left(r_{i}^{(T)}\right)^{q} & \geq\left(r_{i}^{(0)}\right)^{q}\cdot\prod_{t\in T_{lo}:\alpha_{i}^{(t)}>1}\left(\alpha_{i}^{(t)}\right)^{q}=\frac{1}{n}\cdot\prod_{t\in T_{lo}:\alpha_{i}^{(t)}>1}\left(\alpha_{i}^{(t)}\right)^{q}.\label{eq:prod-lb}
\end{align}
Now we can proceed to lower bound this coodinate i.e. we lower bound
the product in (\ref{eq:prod-lb}) using the lower bound we have in
(\ref{eq:sum-lb}).

Intuitively, the worst case behavior i.e. slowest possible increase
in $\left(r_{i}^{(T)}\right)^{q}$ is achieved in one of the two extreme
cases: 

(i) the $\alpha_{i}^{(t)}$ are all minimized i.e. $\left(\alpha_{i}^{(t)}\right)^{q}=\left(1+\epsilon\right)$
in which case $\Theta\left(\frac{1}{\epsilon}\log\left(\frac{n}{\epsilon^{q}}\right)\right)$
such terms are sufficient to make their product $\geq\frac{n}{\epsilon^{q}}$,
which means that we are done, since then we have $\left\Vert r^{(T)}\right\Vert _{q}^{q}\geq\left(r_{i}^{(T)}\right)^{q}\geq\frac{1}{\epsilon^{q}}$;
so setting $\frac{T_{lo}\epsilon^{\frac{q+1}{2}}}{2}\geq\Theta\left(\left(1+\epsilon\right)^{\frac{1}{2q}}\frac{1}{\epsilon}\log\left(\frac{n}{\epsilon^{q}}\right)\right)$
i.e $T_{lo}\geq\Theta\left(\frac{1}{\epsilon^{\frac{q+3}{2}}}\log\left(\frac{n}{\epsilon^{q}}\right)\right)$
is sufficient to make this happen; 

(ii) all the entries are maximized, i.e. $\alpha_{i}^{(t)}=S$ in
which case we have that their product to power $q$ is at least $S^{\frac{qT_{lo}}{S^{1/2}}\frac{\epsilon^{\frac{q+1}{2}}}{2}}\ge\frac{n}{\epsilon^{q}}$,
so if we set $\frac{qT_{lo}}{S^{1/2}}\frac{\epsilon^{\frac{q+1}{2}}}{2}\ln S\geq\log\left(\frac{n}{\epsilon^{q}}\right)$,
ie., $T_{lo}=\Theta\left(\frac{S^{1/2}}{q\ln S}\frac{1}{\epsilon^{\frac{q+1}{2}}}\log\left(\frac{n}{\epsilon^{q}}\right)\right)$,
we guarantee that the corresponding $r_{i}$ increases to a value
larger than $\frac{1}{\epsilon^{q}}$. The fact that these two cases
capture the slowest possible increase is shown in Lemma \ref{lem:prod-lb-formal}.

Therefore we can set
\[
T_{lo}=O\left(\left(\frac{1}{\epsilon}+\frac{S^{1/2}}{q\ln S}\right)\frac{1}{\epsilon^{\frac{q+1}{2}}}\log\left(\frac{n}{\epsilon^{q}}\right)\right).
\]
\end{proof}

\section{Proof of Theorem \ref{thm:high-precision}\protect\label{sec:Proof-of-Theorem-2}}

First, we give guarantee for the subproblem solver (Algorithm \ref{Alg:regularized-regression-1},
proof follows subsequently) .
\begin{lem}
\label{lem:residual-solver-guarantee}For $p\ge1$, $\kappa=\begin{cases}
1 & \text{if }p\le\frac{\log n}{\log n-1}\\
q & \text{otherwise}
\end{cases}$, Algorithm \ref{Alg:regularized-regression-1} either returns $x$
such that $Ax=b$, $\left\Vert x\right\Vert _{2p}\leq2M$ and $\left\langle \theta,x^{2}\right\rangle \le\min_{x:Ax=b}\left\Vert x^{2}\right\Vert _{p}+\left\langle \theta,x^{2}\right\rangle $
or certifies that $\min_{x:Ax=b}\left\Vert x^{2}\right\Vert _{p}+\left\langle \theta,x^{2}\right\rangle \ge\frac{M^{2}}{2\kappa}$
in $O\left(n^{\frac{1}{2q+1}}\right)$ calls to solve a linear system
of the form $ADA^{\top}\phi=b$, where $D$ is an arbitrary non-negative
diagonal matrix.
\end{lem}
The next lemma provides guarantees on the iterate progress in the
main algorithm (Algorithm \ref{alg:main-algo}).
\begin{lem}
\label{lem:high-prec-invariant}For $p\ge2$ $\kappa=\begin{cases}
1 & \text{if }p\le\frac{2\log n}{\log n-1}\\
\frac{p}{p-2} & \text{otherwise}
\end{cases},$ Algorithm \ref{alg:main-algo} maintains that $\left\Vert x^{(t)}\right\Vert _{p}^{p}-\left\Vert x^{*}\right\Vert _{p}^{p}\le16pM^{(t)}$
and that if $x^{(t+1)}\neq x^{(t)}$ then 
\begin{align*}
\left\Vert x^{(t+1)}\right\Vert _{p}^{p}-\left\Vert x^{*}\right\Vert _{p}^{p} & \le\left(1-\frac{1}{2^{13}p\kappa}\right)\left(\left\Vert x^{(t)}\right\Vert _{p}^{p}-\left\Vert x^{*}\right\Vert _{p}^{p}\right).
\end{align*}
\end{lem}
Finally, we show the proof of Theorem \ref{thm:high-precision}.

\begin{proof}
Algorithm \ref{alg:main-algo} terminates when $M^{(t)}\le\frac{\epsilon}{16p\left(1+\epsilon\right)}\left\Vert x^{(t)}\right\Vert _{p}^{p}$.
This gives $\left\Vert x^{(t)}\right\Vert _{p}^{p}-\left\Vert x^{*}\right\Vert _{p}^{p}\le\frac{\epsilon}{16p\left(1+\epsilon\right)}\left\Vert x^{(t)}\right\Vert _{p}^{p}$,
which implies $\left\Vert x^{(t)}\right\Vert _{p}^{p}\le(1+\epsilon)\left\Vert x^{*}\right\Vert _{p}^{p}$
and thus $\left\Vert x^{(t)}\right\Vert _{p}\le(1+\epsilon)\left\Vert x^{*}\right\Vert _{p}$.
Hence, $x^{(t)}$ is a $(1+\epsilon)$ approximate solution. Since
$\frac{\epsilon}{16p\left(1+\epsilon\right)}\left\Vert x^{(t)}\right\Vert _{p}^{p}\ge\frac{\epsilon}{16p\left(1+\epsilon\right)}\left\Vert x^{*}\right\Vert _{p}^{p}$,
the number of times $M^{(t)}$ can be reduced is $O\left(\log\frac{\left\Vert x^{(0)}\right\Vert _{p}^{p}}{\epsilon\left\Vert x^{*}\right\Vert _{p}^{p}}\right)=O\left(p\log\frac{n}{\epsilon}\right)$.
By Lemma \ref{lem:high-prec-invariant}, the number of times the iterate
makes progress is $O\left(2^{13}p\kappa\log\frac{\left\Vert x^{(0)}\right\Vert _{p}^{p}-\left\Vert x^{*}\right\Vert _{p}^{p}}{\epsilon\left\Vert x^{*}\right\Vert _{p}^{p}}\right)=O\left(p^{2}\log n\log\frac{n}{\epsilon}\right)$
where $\kappa=O(\log n)$. Therefore the total number of calls to
the subroutine solver is $O\left(p^{2}\log n\log\frac{n}{\epsilon}\right)$.
By lemma \ref{lem:residual-solver-guarantee}, the subroutine solver
makes $O\left(n^{\frac{1}{2q+1}}\right)=O\left(n^{\frac{p-2}{3p-2}}\right)$
calls to a linear system solver. This concludes the proof.
\end{proof}

\subsection{Proof of Lemma \ref{lem:residual-solver-guarantee}}

We let $\opt=\min_{x:Ax=b}\left\Vert x^{2}\right\Vert _{p}+\left\langle \theta,x^{2}\right\rangle $
and $x^{*}=\arg\min_{x:Ax=b}\left\Vert x^{2}\right\Vert _{p}+\left\langle \theta,x^{2}\right\rangle $.
We consider two cases: when $p\le\frac{\log n}{\log n-1}$ and when
$p>\frac{\log n}{\log n-1}$. We will prove for each case using the
following lemmas:
\begin{lem}
\label{lem:small-p}For $1\le p\le\frac{\log n}{\log n-1}$, Algorithm
\ref{Alg:regularized-regression-1} either returns $x$ such that
$Ax=b$, $\left\Vert x\right\Vert _{2p}\leq2M$ and $\left\langle \theta,x^{2}\right\rangle \le\opt$
or certifies that $\opt\ge\frac{M^{2}}{2}$ in $O(1)$ call to solve
a linear system.
\end{lem}
\begin{lem}
\label{lem:large-p}For $p>\frac{\log n}{\log n-1}$, Algorithm \ref{Alg:regularized-regression-1}
either returns $x$ such that $Ax=b$, $\left\Vert x\right\Vert _{2p}\leq2M$
and $\left\langle \theta,x^{2}\right\rangle \le\opt$ or certifies
that $\opt\ge\frac{M^{2}}{2q}$ in $O\left(n^{\frac{1}{2q+1}}\right)$
calls to solve a linear system.
\end{lem}
To start, we have the following lemma that controls the $\ell_{2}$
term in the objective
\begin{lem}
\label{lem:l2-part}For $r$ such that $\left\Vert r\right\Vert _{q}\le1$,
suppose $x=\arg\min_{x:Ax=b}\langle r+\theta,x^{2}\rangle$. Then
we have $\left\langle \theta,x^{2}\right\rangle \le\opt.$
\end{lem}
\begin{proof}
For $r$ with $\left\Vert r\right\Vert _{q}\le1$, we have
\begin{align*}
\left\langle \theta,x^{2}\right\rangle  & \le\langle r+\theta,x^{2}\rangle\le\langle r+\theta,(x^{*})^{2}\rangle\quad\text{(by definition of \ensuremath{x})}\\
 & \le\left\Vert (x^{*})^{2}\right\Vert _{p}+\left\langle \theta,(x^{*})^{2}\right\rangle \le\opt.
\end{align*}
\end{proof}

Now, let us turn to the first case when $1\le p\le\frac{\log n}{\log n-1}$.
We give the proof for Lemma \ref{lem:small-p}.

\begin{proof}[Proof of Lemma \ref{lem:small-p}]
When $1\le p\le\frac{\log n}{\log n-1}$, we have $q=\frac{p}{p-1}\ge\log n$.
Algorithm \ref{Alg:regularized-regression-1} computes 
\begin{align*}
\widehat{x} & =\min_{x:Ax=b}\left\langle r+\theta,x^{2}\right\rangle 
\end{align*}
where $r_{i}=n^{-\frac{1}{q}}$ for all $i$.

Since $\left\Vert r\right\Vert _{q}=1$, if $\left\Vert \widehat{x}\right\Vert _{2p}\le2M$,
by Lemma \ref{lem:l2-part}, we immediately have $\left\Vert \widehat{x}\right\Vert _{2p}\leq2M$
and $\left\langle \theta,x^{2}\right\rangle \le\opt$. 

Assume that $\left\Vert \widehat{x}\right\Vert _{2p}>2M$. We have
\begin{align*}
\opt & =\left\Vert \left(x^{*}\right)^{2}\right\Vert _{p}+\left\langle \theta,\left(x^{*}\right)^{2}\right\rangle \ge\left\langle r,\left(x^{*}\right)^{2}\right\rangle +\left\langle \theta,\left(x^{*}\right)^{2}\right\rangle \\
 & =\left\langle \theta+r,\left(x^{*}\right)^{2}\right\rangle \ge\left\langle \theta+r,\left(\widehat{x}\right)^{2}\right\rangle \\
 & \ge\frac{1}{n^{\frac{1}{q}}}\left\Vert \widehat{x}^{2}\right\Vert _{1}\ge\frac{1}{n^{\frac{1}{q}}}\left\Vert \widehat{x}^{2}\right\Vert _{p}\qquad\qquad\text{(since \ensuremath{\left\Vert \widehat{x}^{2}\right\Vert _{1}\ge\left\Vert \widehat{x}^{2}\right\Vert _{p}})}\\
 & \ge\frac{1}{2}\left\Vert \widehat{x}\right\Vert _{2p}^{2}\qquad\qquad\text{(since \ensuremath{q\ge\log n})}\\
 & \ge2M^{2}\ge\frac{M^{2}}{2}.
\end{align*}
\end{proof}

For the case when $p>\frac{\log n}{\log n-1}$, the proof for Lemma
\ref{lem:large-p} follows similarly to the analysis of Algorithm
\ref{Alg:low-precision}. We proceed by showing the following invariant.
\begin{lem}[Invariant]
\label{lem:invariant-1}For all $t$, we have that if $\gamma^{(t)}\neq1$
then $\frac{\mathcal{E}(r^{(t+1)}+\theta)-\mathcal{E}(r^{(t)}+\theta)}{\left\Vert r^{(t+1)}\right\Vert _{q}-\left\Vert r^{(t)}\right\Vert _{q}}\ge M^{2}$.
\end{lem}
\begin{proof}
Using Lemma \ref{lem:energy-inc-basic-1} we have
\begin{align*}
\frac{\mathcal{E}(r^{(t+1)}+\theta)-\mathcal{E}(r^{(t)}+\theta)}{\left\Vert r^{(t+1)}\right\Vert _{q}-\left\Vert r^{(t)}\right\Vert _{q}} & \ge\frac{q\cdot\left\Vert r^{(t)}\right\Vert _{q}^{q-1}\left(\sum_{i}\left(r_{i}^{(t)}+\theta_{i}\right)\left(x_{i}^{(t)}\right)^{2}\left(1-\frac{r_{i}^{(t)}+\theta_{i}}{r_{i}^{(t+1)}+\theta_{i}}\right)\right)}{\sum_{i}\left(r_{i}^{(t+1)}\right)^{q}-\left(r_{i}^{(t)}\right)^{q}}\\
 & =\frac{q\cdot\left\Vert r^{(t)}\right\Vert _{q}^{q-1}\left(\sum_{i}\left(x_{i}^{(t)}\right)^{2}\frac{r_{i}^{(t)}+\theta_{i}}{r_{i}^{(t+1)}+\theta_{i}}\left(r_{i}^{(t+1)}-r_{i}^{(t)}\right)\right)}{\sum_{i}\left(r_{i}^{(t+1)}\right)^{q}-\left(r_{i}^{(t)}\right)^{q}}\\
 & \ge\frac{q\cdot\left\Vert r^{(t)}\right\Vert _{q}^{q-1}\left(\sum_{i}\left(x_{i}^{(t)}\right)^{2}\frac{r_{i}^{(t)}}{r_{i}^{(t+1)}}\left(r_{i}^{(t+1)}-r_{i}^{(t)}\right)\right)}{\sum_{i}\left(r_{i}^{(t+1)}\right)^{q}-\left(r_{i}^{(t)}\right)^{q}}\\
 & =\frac{q\cdot\left\Vert r^{(t)}\right\Vert _{q}^{q-1}\left(\sum_{i,\alpha_{i}^{(t)}>1}\left(x_{i}^{(t)}\right)^{2}\frac{r_{i}^{(t)}}{r_{i}^{(t+1)}}\left(r_{i}^{(t+1)}-r_{i}^{(t)}\right)\right)}{\sum_{i,\alpha_{i}^{(t)}>}\left(r_{i}^{(t+1)}\right)^{q}-\left(r_{i}^{(t)}\right)^{q}},
\end{align*}
where in the second inequality we use $\frac{r_{i}^{(t)}+\theta_{i}}{r_{i}^{(t+1)}+\theta_{i}}\ge\frac{r_{i}^{(t)}}{r_{i}^{(t+1)}}$
for $r_{i}^{(t+1)}\ge r_{i}^{(t)}$, $\theta\ge0$. For $i$ such
that $\alpha_{i}^{(t)}>1$, we have $r_{i}^{(t+1)}=\alpha_{i}^{(t)}r_{i}^{(t)}$,
thus
\begin{align*}
\frac{q\cdot\left\Vert r^{(t)}\right\Vert _{q}^{q-1}\left(x_{i}^{(t)}\right)^{2}\frac{r_{i}^{(t)}}{r_{i}^{(t+1)}}\left(r_{i}^{(t+1)}-r_{i}^{(t)}\right)}{\left(r_{i}^{(t+1)}\right)^{q}-\left(r_{i}^{(t)}\right)^{q}} & =\gamma_{i}^{(t)}M^{2}\cdot\frac{q\left(1-\frac{1}{\alpha_{i}^{(t)}}\right)}{\left(\alpha_{i}^{(t)}\right)^{q}-1}\\
 & \ge\gamma_{i}^{(t)}M^{2}\cdot\frac{1}{\left(\alpha_{i}^{(t)}\right)^{q}}\\
 & =M^{2},
\end{align*}
where the first inequality is due to $\frac{q\left(\alpha-1\right)}{\alpha\left(\alpha^{q}-1\right)}\ge\frac{1}{\alpha^{q}}$,
for $\alpha>1$. We can then obtain the desired conclusion from here.
\end{proof}

\begin{lem}[Case 1]
\label{lem:high-case1}Let $r$ be a dual solution and $x=\arg\min_{\widehat{x}:A\widehat{x}=b}\langle r+\theta,\widehat{x}^{2}\rangle$.
If $\left\Vert \left\Vert r\right\Vert _{q}^{q-1}\cdot\frac{x^{2}}{r^{q-1}}\right\Vert _{\infty}\leq2M$
then $\left\Vert x\right\Vert _{2p}\leq2M$ and $\left\langle \theta,x^{2}\right\rangle \le\opt.$
\end{lem}
\begin{proof}
If
\begin{align*}
\left\Vert \left\Vert r\right\Vert _{q}^{q-1}\cdot\frac{x^{2}}{r^{q-1}}\right\Vert _{\infty} & \leq2M^{2},
\end{align*}
for all $i$ we have 
\begin{align*}
x_{i}^{2} & \le4M^{2}\frac{r_{i}^{q-1}}{\left\Vert r\right\Vert _{q}^{q-1}},
\end{align*}
which gives 
\begin{align*}
x_{i}^{2p} & \le2^{2p}M^{2p}\frac{r_{i}^{q}}{\left\Vert r\right\Vert _{q}^{q}},
\end{align*}
We obtain 
\begin{align*}
\left\Vert x\right\Vert _{2p}^{2p} & \le2^{2p}M^{2p},
\end{align*}
as needed. The second claim comes directly from Lemma \ref{lem:l2-part}.
\end{proof}

\begin{lem}[Case 3]
\label{lem:high-dual-solution}If the algorithm returns $r^{(T)}$,
then $\mathcal{E}\left(\frac{r^{(T)}}{\left\Vert r^{(T)}\right\Vert _{q}}+\theta\right)\geq\frac{M^{2}}{2q}.$
\end{lem}
\begin{proof}
We have that
\begin{align*}
\frac{\mathcal{E}(r^{(T)}+\theta)}{\left\Vert r^{(T)}\right\Vert _{q}} & =\frac{\mathcal{E}(r^{(0)}+\theta)+\sum_{t=0}^{T-1}\left(\mathcal{E}(r^{(t+1)}+\theta)-\mathcal{E}(r^{(t)}+\theta)\right)}{\left\Vert r^{(T)}\right\Vert _{q}}\\
 & \geq\frac{\sum_{t=0}^{T-1}\left(\left\Vert r^{(t+1)}\right\Vert _{q}-\left\Vert r^{(t)}\right\Vert _{q}\right)\cdot M^{2}}{\left\Vert r^{(T)}\right\Vert _{q}}\quad\hfill\text{(due to the invariant)}\\
 & \geq\frac{\left(\left\Vert r^{(T)}\right\Vert _{q}-\left\Vert r^{(0)}\right\Vert _{q}\right)\cdot M^{2}}{\left\Vert r^{(T)}\right\Vert _{q}}\\
 & =M^{2}\cdot\left(1-\frac{\frac{2q-1}{2q}}{\left\Vert r^{(T)}\right\Vert _{q}}\right)\quad\hfill\text{(since \ensuremath{\left\Vert r^{(0)}\right\Vert _{q}=\frac{2q-1}{2q}})}\\
 & =\frac{M^{2}}{2q}\quad\hfill\text{(since \ensuremath{\left\Vert r^{(T)}\right\Vert _{q}\ge1})}.
\end{align*}
Finally since $\left\Vert r^{(T)}\right\Vert _{q}\ge1$
\begin{align*}
\mathcal{E}\left(\frac{r^{(T)}}{\left\Vert r^{(T)}\right\Vert _{q}}+\theta\right) & \ge\frac{\mathcal{E}(r^{(T)}+\theta)}{\left\Vert r^{(T)}\right\Vert _{q}}\ge\frac{M^{2}}{2q}.
\end{align*}
\end{proof}

\paragraph{Convergence Analysis}

We run the algorithm for $T$ iterations. The algorithm terminates
if at any point it finds a solution $x$ that satisfies the desired
bound (otherwise it is unable to further perturb the dual solution).
Otherwise, we show that it must finish very fast.

Suppose we run it for $T=T_{hi}+T_{lo}$ iterations. Let the iterations
in $T_{hi}$ correspond to those where at least a single $r_{i}$
was scaled by $\geq S=n^{\frac{2}{2q+1}}$. Let $T_{lo}$ be the remaining
iterations.
\begin{lem}
\label{lem:bound-t-hi-1}We have $T_{hi}\leq\frac{2n}{S^{q}}$.
\end{lem}
\begin{proof}
Suppose the contrary. Then we claim that these perturbations alone
will have increased $r$ a lot to the point where $\left\Vert r\right\Vert _{q}^{q}\geq1$.
Indeed, let $r_{i}$ be the current value of coordinate $i$ and $r_{i}'$
be its value after being increased, and assume that $\frac{r_{i}'}{r_{i}}\ge S$.
Since $r$ is initialized to $\frac{2q-1}{2q}\frac{1}{n^{1/q}}$,
in the worst case each perturbation in $T_{hi}$ touches a different
$i$. Therefore this establishes a lower bound of $T_{hi}\cdot\frac{S^{q}}{n}\left(\frac{2q-1}{2q}\right)^{q}\ge T_{hi}\cdot\frac{S^{q}}{2n}$
on $\left\Vert r\right\Vert _{q}^{q}$. As this must be at most $1$,
since otherwise we obtained a good solution per Lemma \ref{lem:high-dual-solution},
we obtain the conclusion.
\end{proof}

Now we claim that we can either look at the history produced in $T_{lo}$
and obtain an approximately feasible solution, or a single coordinate
$r_{i}$ must have increased a lot.
\begin{lem}
\label{lem-t-lo-cert-1}Consider the set of iterates $(r^{(t)},x^{(t)})$
used for the iterates in $T_{lo}$. If 
\[
\left\Vert \frac{1}{T_{lo}}\sum_{t\in T_{lo}}x^{(t)}\right\Vert _{2p}>2M
\]
then there exists a coordinate $i$ for which 
\[
\sum_{t\in T_{lo}:\alpha_{i}^{(t)}>1}\sqrt{\alpha_{i}^{(t)}}\geq\frac{T_{lo}}{4}.
\]
\end{lem}
\begin{proof}
Suppose $\left\Vert \frac{1}{T_{lo}}\sum_{t\in T_{lo}}x^{(t)}\right\Vert _{2p}>2M$.
Note that by the update rule,
\begin{align*}
\frac{x_{i}^{(t)}}{M} & \le\sqrt{2}\sqrt{\frac{\left(r_{i}^{(t)}\right)^{q-1}}{\left\Vert r^{(t)}\right\Vert _{q}^{q-1}}}+\boldsymbol{1}_{\alpha_{i}>1}\sqrt{\frac{\alpha_{i}^{(t)q}\left(r_{i}^{(t)}\right)^{q-1}}{\left\Vert r^{(t)}\right\Vert _{q}^{q-1}}}
\end{align*}
Hence we can write 
\begin{align*}
\left\Vert \sum_{t\in T_{lo}}\frac{x^{(t)}}{M}\right\Vert _{2p} & \le\left\Vert \sqrt{2}\sum_{t\in T_{lo}}\sqrt{\frac{\left(r^{(t)}\right)^{q-1}}{\left\Vert r^{(t)}\right\Vert _{q}^{q-1}}}+\overrightarrow{\left(\sum_{t\in T_{lo},\alpha_{i}^{(t)}>1}\sqrt{\frac{\alpha_{i}^{(t)q}\left(r_{i}^{(t)}\right)^{q-1}}{\left\Vert r^{(t)}\right\Vert _{q}^{q-1}}}\right)_{i}}\right\Vert _{2p}\\
 & \le\sqrt{2}\sum_{t\in T_{lo}}\left\Vert \sqrt{\frac{\left(r^{(t)}\right)^{q-1}}{\left\Vert r^{(t)}\right\Vert _{q}^{q-1}}}\right\Vert _{2p}+\left\Vert \overrightarrow{\left(\sum_{t\in T_{lo},\alpha_{i}^{(t)}>1}\sqrt{\frac{\alpha_{i}^{(t)q}\left(r_{i}^{(t)}\right)^{q-1}}{\left\Vert r^{(t)}\right\Vert _{q}^{q-1}}}\right)_{i}}\right\Vert _{2p}\\
 & \hfill\text{ (by triangle inequality)}\\
 & =\sqrt{2}T_{lo}+\left\Vert \overrightarrow{\left(\sum_{t\in T_{lo},\alpha_{i}^{(t)}>1}\sqrt{\frac{\alpha_{i}^{(t)q}\left(r_{i}^{(t)}\right)^{q-1}}{\left\Vert r^{(t)}\right\Vert _{q}^{q-1}}}\right)_{i}}\right\Vert _{2p}.
\end{align*}
We obtain 
\begin{align*}
\left\Vert \overrightarrow{\left(\sum_{t\in T_{lo},\alpha_{i}^{(t)}>1}\sqrt{\frac{\alpha_{i}^{(t)q}\left(r_{i}^{(t)}\right)^{q-1}}{\left\Vert r^{(t)}\right\Vert _{q}^{q-1}}}\right)_{i}}\right\Vert _{2p} & \ge\left(2-\sqrt{2}\right)T_{lo}\ge\frac{T_{lo}}{2}
\end{align*}
On the other hand, we have 
\begin{align*}
\sum_{i}\left(\sum_{t\in T_{lo},\alpha_{i}^{(t)}>1}\sqrt{\frac{\alpha_{i}^{(t)q}\left(r_{i}^{(t)}\right)^{q-1}}{\left\Vert r^{(t)}\right\Vert _{q}^{q-1}}}\right)^{2p} & =\sum_{i}\left(\sum_{t\in T_{lo},\alpha_{i}^{(t)}>1}\sqrt{\frac{\alpha_{i}^{(t)}\left(r_{i}^{(t+1)}\right)^{q-1}}{\left\Vert r^{(t)}\right\Vert _{q}^{q-1}}}\right)^{2p}\\
\le\sum_{i}\frac{\left(r_{i}^{(T)}\right)^{q}}{\left\Vert r^{(0)}\right\Vert _{q}^{q}}\left(\sum_{t\in T_{lo},\alpha_{i}^{(t)}>1}\sqrt{\alpha_{i}^{(t)}}\right)^{2p} & \le\frac{\left\Vert r^{(T)}\right\Vert _{q}^{q}}{\left\Vert r^{(0)}\right\Vert _{q}^{q}}\max_{i}\left(\sum_{t\in T_{lo},\alpha_{i}^{(t)}>1}\sqrt{\alpha_{i}^{(t)}}\right)^{2p}\\
\le\left(\frac{2q}{2q-1}\right)^{q}\max_{i}\left(\sum_{t\in T_{lo},\alpha_{i}^{(t)}>1}\sqrt{\alpha_{i}^{(t)}}\right)^{2p} & \qquad\text{(since \ensuremath{\left\Vert r^{(0)}\right\Vert _{q}=\frac{2q}{2q-1}})}\\
\le2\max_{i}\left(\sum_{t\in T_{lo},\alpha_{i}^{(t)}>1}\sqrt{\alpha_{i}^{(t)}}\right)^{2p}, & \qquad\text{(since \ensuremath{q\ge1})}
\end{align*}
\end{proof}
Therefore there exists $i$ such that 
\begin{align*}
\left(\sum_{t\in T_{lo},\alpha_{i}^{(t)}>1}\sqrt{\alpha_{i}^{(t)}}\right)^{2p} & \ge\frac{1}{2}\left(\frac{T_{lo}}{2}\right)^{2p},
\end{align*}
which gives us
\begin{align*}
\sum_{t\in T_{lo},\alpha_{i}^{(t)}>1}\sqrt{\alpha_{i}^{(t)}} & \ge\frac{T_{lo}}{2}\frac{1}{2^{\frac{1}{2p}}}\ge\frac{T_{lo}}{4},\mbox{ since }p\ge1.
\end{align*}

This lemma enables us to upper bound $T_{lo}$.

\begin{lem}
\label{lem:bound-t-lo-1}We have $T_{lo}\leq\Theta\left(\frac{S^{1/2}}{\ln S}\ln n+\ln n\right)$.
\end{lem}
\begin{proof}
From Lemma \ref{lem-t-lo-cert-1} we know that there exists a coordinate
$i$ for which 
\begin{equation}
\sum_{t\in T_{lo}:\alpha_{i}^{(t)}>1}\sqrt{\alpha_{i}^{(t)}}>\frac{T_{lo}}{4}.\label{eq:sum-lb-1}
\end{equation}
Furthermore by definition for all iterates in $T_{lo}$ we have that
pointwise $\alpha_{i}^{(t)}=\frac{r_{i}^{(t+1)}}{r_{i}^{(t)}}\leq S$
and $\alpha_{i}^{(t)}=\left(\gamma_{i}^{(t)}\right)^{1/q}\ge2^{\frac{1}{q}}$.
This enables us to lower bound the final value of $\left(r_{i}^{(T)}\right)^{q}$
which is a lower bound on $\left\Vert r^{(T)}\right\Vert _{q}^{q}$.
More precisely, we have $\frac{r_{i}^{(t+1)}}{r_{i}^{(t)}}\ge\alpha_{i}^{(t)}$
thus
\begin{align}
\left(r_{i}^{(T)}\right)^{q} & \geq\left(r_{i}^{(0)}\right)^{q}\cdot\prod_{t\in T_{lo}:\alpha_{i}^{(t)}>1}\left(\alpha_{i}^{(t)}\right)^{q}=\frac{2q-1}{2q}\cdot\frac{1}{n}\cdot\prod_{t\in T_{lo}:\alpha_{i}^{(t)}>1}\left(\alpha_{i}^{(t)}\right)^{q}.\label{eq:prod-lb-1}
\end{align}
Now we can proceed to lower bound this $r_{i}$ i.e. we lower bound
the product in (\ref{eq:prod-lb-1}) using the lower bound we have
in (\ref{eq:sum-lb-1}).

Similarly to the previous section, the worst case behavior i.e. slowest
possible increase in $\left(r_{i}^{(T)}\right)^{q}$ is achieved in
one of the two extreme cases: 

(i) the $\alpha_{i}^{(t)}$ are all minimized i.e. $\alpha_{i}^{(t)}=2^{\frac{1}{q}}$
in which case $\Theta\left(\ln n\right)$ such terms are sufficient
to make their product $\geq2n\ge\frac{2qn}{2q-1}$, which means that
we are done, since then we have $\left\Vert r^{(T)}\right\Vert _{q}^{q}\geq\left(r_{i}^{(T)}\right)^{q}\geq1$;
so setting $T_{lo}\geq\Theta\left(\ln n\right)$ is sufficient to
make this happen; 

(ii) all the entries are maximized, i.e. $\alpha_{i}^{(t)}=S$ in
which case we have that their product to power $q$ is at least $S^{\frac{T_{lo}q}{4S^{1/2}}}\ge2n\ge\frac{2qn}{2q-1}$,
so if we set $\frac{T_{lo}q}{4S^{1/2}}\ln S\geq\ln2n$, ie, $T_{lo}\ge\frac{8S^{1/2}\ln(n)}{q\ln S}$,
we guarantee that the corresponding $r_{i}$ increases to a value
larger than $2$. The fact that these two cases capture the slowest
possible increase is shown in Lemma \ref{lem:prod-lb-formal}.

Therefore we can set
\[
T_{lo}=O\left(\frac{S^{1/2}}{\ln S}\ln n+\ln n\right).
\]
\end{proof}

Finally, by the choice $S=n^{\frac{2}{2q+1}}$, we obtain the runtime
guarantee.
\begin{lem}
\label{lem:high-subsolver}Algorithm \ref{Alg:regularized-regression-1}
terminates in $O\left(n^{\frac{1}{2q+1}}\right)$ iterations.
\end{lem}
\begin{proof}[Proof of Lemma \ref{lem:large-p}]
The proof of Lemma \ref{lem:residual-solver-guarantee} immediately
follows from Lemmas \ref{lem:high-case1}, \ref{lem:high-dual-solution}
and \ref{lem:high-subsolver}.
\end{proof}

\subsection{Proof of Lemma \ref{lem:high-prec-invariant}}

\begin{proof}[Proof of Lemma \ref{lem:high-prec-invariant}]
We define the function $\res_{x}$ as follows
\begin{align*}
\res_{x}\left(\Delta\right) & =\left\langle g,\Delta\right\rangle -\left\langle R,\Delta^{2}\right\rangle -\left\Vert \Delta\right\Vert _{p}^{p}
\end{align*}
where $g=\left|x\right|^{p-2}x$, $R=2\left|x\right|^{p-2}$. We use
the following property of this function from \citet{adil2019iterative,adil2022fast}:
For $\lambda=16p$ and for all $\Delta$
\begin{align}
\left\Vert x\right\Vert _{p}^{p}-\left\Vert x-\frac{\Delta}{p}\right\Vert _{p}^{p} & \ge\res_{x}\left(\Delta\right);\label{eq:res-prop1}\\
\left\Vert x\right\Vert _{p}^{p}-\left\Vert x-\lambda\frac{\Delta}{p}\right\Vert _{p}^{p} & \le\lambda\res_{x}\left(\Delta\right).\label{eq:res-prop2}
\end{align}

We prove the claim by induction.

For $t=0$, we have $M^{(0)}:=\frac{\left\Vert x^{(0)}\right\Vert _{p}^{p}}{16p}\ge\frac{\left\Vert x^{(t)}\right\Vert _{p}^{p}-\left\Vert x^{*}\right\Vert _{p}^{p}}{16p}$.

Now assume that we have $\left\Vert x^{(t)}\right\Vert _{p}^{p}-\left\Vert x^{*}\right\Vert _{p}^{p}\le16pM^{(t)}$.
We have two cases.

Case 1. $\ressolver$ returns an infeasibility certificate or $\ressolver$
returns a primal solution $\tilde{\Delta}$ such that $\left\langle R^{(t)},\tilde{\Delta}^{2}\right\rangle \ge2M^{(t)}$.
In both scenarios, using Lemma \ref{lem:residual-solver-guarantee}
we have 
\begin{align*}
\min_{\substack{A\Delta=0\\
\left\langle g^{(t)},\Delta\right\rangle =\frac{M^{(t)}}{2}
}
}\left\Vert \Delta^{2}\right\Vert _{\frac{p}{2}}+(M^{(t)})^{\frac{2-p}{p}}\left\langle R^{(t)},\Delta^{2}\right\rangle  & \ge2(M^{(t)})^{\frac{2}{p}}.
\end{align*}
Hence for all $\Delta$ such that $A\Delta=0$, $\left\langle g^{(t)},\Delta\right\rangle =\frac{M^{(t)}}{2}$,
either $\left\Vert \Delta^{2}\right\Vert _{\frac{p}{2}}\ge(M^{(t)})^{\frac{2}{p}}\Leftrightarrow\left\Vert \Delta\right\Vert _{p}^{p}\ge M^{(t)}$
or $(M^{(t)})^{\frac{2-p}{p}}\left\langle R^{(t)},\Delta^{2}\right\rangle \ge(M^{(t)})^{\frac{2}{p}}\Leftrightarrow\left\langle R^{(t)},\Delta^{2}\right\rangle \ge M^{(t)}$.
For all $\Delta$ such that $A\Delta=0$, we can write $\left\langle g^{(t)},\Delta\right\rangle =a\frac{M^{(t)}}{2}$,
for some constant $a\in\R$. We obtain either $\left\Vert \Delta\right\Vert _{p}^{p}\ge a^{p}M^{(t)}$
or $\left\langle R^{(t)},\Delta^{2}\right\rangle \ge a^{2}M^{(t)}$,
and thus for all $\Delta$
\begin{align*}
\res_{x^{(t)}}\left(\Delta\right) & \le M^{(t)}\left(\frac{1}{2}a-\min\left\{ a^{2},a^{p}\right\} \right)\le\frac{M^{(t)}}{2}=M^{(t+1)}.
\end{align*}
We write $\overline{\Delta}=\frac{x^{(t)}-x^{*}}{\lambda/p}$, for
$\lambda=16p$. Using property (\ref{eq:res-prop2}) of the $\res_{x}$,
we have 
\begin{align*}
\left\Vert x^{(t+1)}\right\Vert _{p}^{p}-\left\Vert x^{*}\right\Vert _{p}^{p} & =\left\Vert x^{(t)}\right\Vert _{p}^{p}-\left\Vert x^{*}\right\Vert _{p}^{p}\\
 & =\left\Vert x^{(t)}\right\Vert _{p}^{p}-\left\Vert x^{(t)}-\lambda\frac{\overline{\Delta}}{p}\right\Vert _{p}\\
 & \le\lambda\res_{x^{(t)}}\left(\overline{\Delta}\right)\\
 & \le16pM^{(t+1)}.
\end{align*}

Case 2. We have $\left\langle R,\tilde{\Delta}^{2}\right\rangle <2M^{(t)}$
and $\left\Vert \tilde{\Delta}\right\Vert _{p}\le4\sqrt{\kappa}(M^{(t)})^{\frac{1}{p}}$
and $\left\langle g,\tilde{\Delta}\right\rangle =\frac{M^{(t)}}{2}$
\begin{align*}
\left\Vert x^{(t)}\right\Vert _{p}^{p}-\left\Vert x^{(t+1)}\right\Vert _{p}^{p} & =\left\Vert x^{(t)}\right\Vert _{p}^{p}-\left\Vert x^{(t)}-\frac{\tilde{\Delta}}{64p\kappa}\right\Vert _{p}^{p}\\
 & \ge\res_{x^{(t)}}\left(\frac{\tilde{\Delta}}{64\kappa}\right)\\
 & =\left\langle g,\frac{\tilde{\Delta}}{64\kappa}\right\rangle -\left\langle R,\left(\frac{\tilde{\Delta}}{64\kappa}\right)^{2}\right\rangle -\left\Vert \frac{\tilde{\Delta}}{64\kappa}\right\Vert _{p}^{p}\\
 & \ge\frac{M^{(t)}}{2^{7}\kappa}-\frac{M^{(t)}}{2^{11}\kappa^{2}}-\frac{M^{(t)}}{2^{4p}\kappa^{\frac{p}{2}}}\\
 & \ge\frac{M^{(t)}}{2^{7}\kappa}-\frac{M^{(t)}}{2^{11}\kappa}-\frac{M^{(t)}}{2^{8}\kappa},\qquad(\mbox{since }p\ge2,\kappa\ge1)\\
 & \ge\frac{M^{(t)}}{2^{9}\kappa}\ge\frac{1}{2^{13}p\kappa}\left(\left\Vert x^{(t)}\right\Vert _{p}^{p}-\left\Vert x^{*}\right\Vert _{p}^{p}\right),
\end{align*}
from which we obtain
\begin{align*}
\left\Vert x^{(t+1)}\right\Vert _{p}^{p}-\left\Vert x^{*}\right\Vert _{p}^{p} & \le\left\Vert x^{(t)}\right\Vert _{p}^{p}-\left\Vert x^{*}\right\Vert _{p}^{p}-\frac{1}{2^{13}p\kappa}\left(\left\Vert x^{(t)}\right\Vert _{p}^{p}-\left\Vert x^{*}\right\Vert _{p}^{p}\right)\\
 & \le\left(1-\frac{1}{2^{13}p\kappa}\right)\left(\left\Vert x^{(t)}\right\Vert _{p}^{p}-\left\Vert x^{*}\right\Vert _{p}^{p}\right)
\end{align*}
as needed.
\end{proof}

\section{Lower Bound Lemma\protect\label{sec:prod-lb-formal}}
\begin{lem}
\label{lem:prod-lb-formal}Let a set of nonnegative reals $\beta_{1},\dots,\beta_{k}$
such that $1+\epsilon\leq\beta_{i}\le S$, and $\sum_{i=1}^{k}\beta_{i}^{\frac{1}{r}}\geq K$,
where $r\geq2$. Then for any $k$ one has that
\[
\prod_{i=1}^{k}\beta_{i}\geq\min\left\{ S^{\frac{K}{S^{1/r}}},(1+\epsilon)^{\frac{K}{(1+\epsilon)^{1/r}}}\right\} .
\]
\end{lem}
\begin{proof}
Consider a fixed $k$, and let us attempt to minimize the product
of $\beta_{i}$\textquoteright s subject to the constraints. W.l.o.g.
we have $\sum_{i=1}^{k}\beta_{i}^{\frac{1}{r}}=K$. Equivalently we
want to minimize $\sum_{i=1}^{k}\log(\beta_{i})$, which is a concave
function. Therefore its minimizer is attained on the boundary of the
feasible domain. This means that for some $0\leq k'\leq k-1$, there
are $k'$ elements equal to $1+\epsilon$, $k-1-k'$ equal to $S$,
and one which is exactly equal to the remaining budget, i.e. $\left(K-k'(1+\epsilon)^{1/r}-(k-1-k')S^{1/r}\right)$,
which yields the product $(1+\epsilon)^{k'}S^{k-k'-1}\cdot\left(K-k'(1+\epsilon)^{1/r}-(k-1-k')S^{1/r}\right)$.
This can be relaxed by allowing $k$ and $k'$ to be non-integral.
Hence we aim to minimize the product $(1+\epsilon)^{k'}S^{k-k'-1}$
subject to $k'(1+\epsilon)^{1/r}-(k-1-k')S^{1/r}=K$.

Finally, we observe that we can always obtain a better solution by
placing all the available mass on a single one of the factors, i.e.
we lower bound either by $S^{\frac{K}{S^{1/r}}}$ or $(1+\epsilon)^{\frac{K}{(1+\epsilon)^{1/r}}}$,
whichever is lowest.

\section{Iterative Refinement}

In this section we provide a general technique for solving optimization
problems to high-precision, by reducing then to an adaptive sequence
of easier optimization problems, which only require approximate solutions.
This formalizes the minimal requirements for the iterative refinement
scheme employed in \citet{adil2019iterative,adil2019fast} to go through.
We state the main lemma below.
\begin{lem}
\label{lem:general-iterative-refinement}Let $\mathcal{D}\subseteq\mathbb{R}^{n}$
be a convex set, and let $f:\mathcal{D}\rightarrow\mathbb{R}$ be
a convex function. Let $\eta\geq0$ be a scalar, and suppose that
for any $x\in\mathcal{D}$ there exists a function $h_{x}$ that approximates
the Bregman divergence at $x$ in the sense that 
\[
\frac{1}{\eta}h_{x}\left(\eta\delta\right)\leq f\left(x+\delta\right)-f\left(x\right)-\left\langle \nabla f\left(x\right),\delta\right\rangle \leq h_{x}\left(\delta\right)\,.
\]
Given access to an oracle that for any direction $v$ can provide
$\kappa$-approximate minimizers to $\left\langle v,\delta\right\rangle +h_{x}\left(\delta\right)$
in the sense that it returns $\delta^{\sharp}$ such that $v+\delta^{\sharp}\in\mathcal{D}$
and 
\[
\left\langle v,\delta^{\sharp}\right\rangle +h_{x}\left(\delta^{\sharp}\right)\leq\frac{1}{\kappa}\left(\min_{v+\delta\in\mathcal{D}}\left\langle v,\delta\right\rangle +h_{x}\left(\delta\right)\right)\,,
\]
along with an initial point $x_{0}\in\mathcal{D}$, in $O\left(\frac{\kappa}{\eta}\ln\frac{f\left(x_{0}\right)-f\left(x^{*}\right)}{\varepsilon}\right)$
calls to the oracle one can obtain a point $x$ such that $f\left(x\right)\leq f\left(x^{*}\right)+\varepsilon$,
where $x^{*}\in\arg\min_{x\in\mathcal{D}}f\left(x\right)$.
\end{lem}
\begin{proof}
Let $\delta^{\sharp}$ be the a $\kappa$-approximate minimizer of
$\left\langle \nabla f\left(x\right),\delta^{\sharp}\right\rangle +h_{x}\left(\delta^{\sharp}\right)$,
which by definition satisfies:
\begin{equation}
\left\langle \nabla f\left(x\right),\delta^{\sharp}\right\rangle +h_{x}\left(\delta^{\sharp}\right)\leq\frac{1}{\kappa}\left(\min_{v+\delta\in\mathcal{D}}\left\langle \nabla f\left(x\right),\delta\right\rangle +h_{x}\left(\delta\right)\right)\,.\label{eq:k-approx-min}
\end{equation}
Updating our iterate to $x'=x+\delta^{\sharp}$ we can bound the new
function value as
\begin{align*}
 & f\left(x+\delta^{\sharp}\right)\\
 & =f\left(x\right)+\left\langle \nabla f\left(x\right),\delta^{\sharp}\right\rangle +h_{x}\left(\delta^{\sharp}\right)\tag{Bregman divergence upper bound}\\
 & \leq f\left(x\right)+\frac{\eta}{\kappa}\left(\left\langle \nabla f\left(x\right),x^{*}-x\right\rangle +\frac{1}{\eta}h_{x}\left(\eta\left(x^{*}-x\right)\right)\right)\tag{using (\ref{eq:k-approx-min})}\\
 & =f\left(x\right)+\frac{\eta}{\kappa}\left(\left\langle \nabla f\left(x\right),x^{*}-x\right\rangle +\left(f\left(x^{*}\right)-f\left(x\right)-\left\langle \nabla f\left(x\right),x-x^{*}\right\rangle \right)\right)\tag{Bregman divergence lower bound}\\
 & =f\left(x\right)+\frac{\eta}{\kappa}\left(f\left(x^{*}\right)-f\left(x\right)\right)\,,
\end{align*}
from where we equivalently obtain that 
\[
f\left(x+\delta^{\sharp}\right)-f\left(x^{*}\right)\leq\left(1-\frac{\eta}{\kappa}\right)\left(f\left(x\right)-f\left(x^{*}\right)\right)\,.
\]
Therefore to reduce the initial error $f\left(x_{0}\right)-f\left(x^{*}\right)$
to $\varepsilon$ it suffices to iterate $O\left(\frac{\kappa}{\eta}\ln\frac{f\left(x_{0}\right)-f\left(x^{*}\right)}{\varepsilon}\right)$
times.
\end{proof}
The following lemma provides a sandwiching inequality for the Bregman
divergence of $\left\Vert x\right\Vert _{p}^{p}$.
\begin{lem}
[\citet{adil2019fast}, Lemma B.1]\label{lem:lp-div-sandwich} For
any $x,\delta$ and $p\geq2$, we have for $r=x^{p-2}$ and $g=px^{p-1}$,
\[
\frac{p}{8}\left\langle r,\delta^{2}\right\rangle +\frac{1}{2^{p+1}}\left\Vert \delta\right\Vert _{p}^{p}\leq\left\Vert x+\delta\right\Vert _{p}^{p}-\left\Vert x\right\Vert _{p}^{p}-\left\langle g,\delta\right\rangle \leq2p^{2}\left\langle r,\delta^{2}\right\rangle +p^{p}\left\Vert \delta\right\Vert _{p}^{p}\,.
\]
\end{lem}
As a corollary we see that the function $h_{x}\left(\delta\right)=2p^{2}\left\langle x^{p-2},\delta^{2}\right\rangle +p^{p}\left\Vert \delta\right\Vert _{p}^{p}$
satisfies the inequality required by Lemma \ref{lem:general-iterative-refinement}
for $\eta=\frac{1}{4p}$. We can thus conclude that given access to
an oracle that approximately minimizes mixed $\ell_{2}+\ell_{p}$
regression objectives, one can efficiently generate a high precision
solution.
\begin{cor}
Consider the $\ell_{p}$ regression problem $\min_{f:B^{\top}f=d}\left\Vert f\right\Vert _{p}^{p}$.
Given access to an oracle that can compute $\kappa$-approximate minimizers
to the optimization problem
\[
V^{*}:=\min_{f:B^{\top}\Delta f=0}\left\langle pf^{p-1},\Delta f\right\rangle +2p^{2}\left\langle f^{p-2},\Delta f^{2}\right\rangle +p^{p}\left\Vert \Delta f\right\Vert _{p}^{p}
\]
 in the sense that it returns $\Delta f$ satisfying $B^{\top}\Delta f=0$
and 
\[
\left\langle pf^{p-1},\Delta f\right\rangle +2p^{2}\left\langle f^{p-2},\Delta f^{2}\right\rangle +p^{p}\left\Vert \Delta f\right\Vert _{p}^{p}\leq\frac{1}{\kappa}V^{*}\,,
\]
along with an initial point $f_{0}$, satisfying $B^{\top}f=d$, in
$O\left(\kappa p\ln\frac{\left\Vert f_{0}\right\Vert _{p}^{p}-\left\Vert f^{*}\right\Vert _{p}^{p}}{\varepsilon}\right)$
calls to the oracle one can obtain a point $f$ such that $\left\Vert f\right\Vert _{p}^{p}\leq\left\Vert f^{*}\right\Vert _{p}^{p}+\varepsilon$,
where $f^{*}\in\arg\min_{B^{\top}f=d}\left\Vert f\right\Vert _{p}^{p}$.
\end{cor}
\begin{proof}
Using Lemma \ref{lem:lp-div-sandwich} we verify that the function
$h_{f}\left(\Delta f\right)=2p^{2}\left\langle f^{p-2},\Delta f^{2}\right\rangle +p^{p}\left\Vert \Delta f\right\Vert _{p}^{p}$
satisfies 
\[
\frac{1}{\eta}h_{f}\left(\eta\Delta f\right)\leq\left\Vert f+\Delta f\right\Vert _{p}^{p}-\left\Vert f\right\Vert _{p}^{p}+\left\langle pf^{p-1},\Delta f\right\rangle \leq h_{f}\left(\Delta f\right)
\]
 for $\eta=\frac{1}{4p}$. Therefore by Lemma \ref{lem:general-iterative-refinement}
we can need $O\left(\kappa p\ln\frac{\left\Vert f_{0}\right\Vert _{p}^{p}-\left\Vert f^{*}\right\Vert _{p}^{p}}{\varepsilon}\right)$
iterations to obtain an $\varepsilon$-additive error to the regression
problem.
\end{proof}
\end{proof}

\section{Additional Experimental Results}

\subsection{Data generation}

\paragraph*{Random matrices.}

The entries of $A$ and $b$ are generated uniformly at randomly between
$0$ and $1$. 

\paragraph*{Random graphs.}

We use the procedure in \citet{adil2019fast} to generate random graphs
and the corresponding $A$ and $b$. The generated graph is a weighted
graph, where the vertices are generated by choosing a point in $[0,1]^{10}$
uniformly at random, each vertex is connected to the 10 nearest neighbors.
The edge weights are generated by a gaussian type function (by Flores-Calder-Lerman).
$k$ (around 10) nodes are labeled in $[0,1]$ and let $g$ be the
label vector. Let $B$ be the edge-vertex adjacency matrix, $W$ be
the diagonal matrix with edge weights. We generate $A=W^{1/p}B$,
$b=-B[:,n:n+k]g$.

\subsection{Correctness of solution}

In Figure \ref{fig:error}, we plot the error of the solutions outputted
by our algorithm and $p$-IRLS against CVX in the random matrices
and random graphs instances for $\epsilon=10^{-10}$. In all cases,
the error is below $\epsilon$.

\begin{figure*}
\subfloat[matrix size=$n\times(n-50),p=8$]{\includegraphics[width=0.25\textwidth]{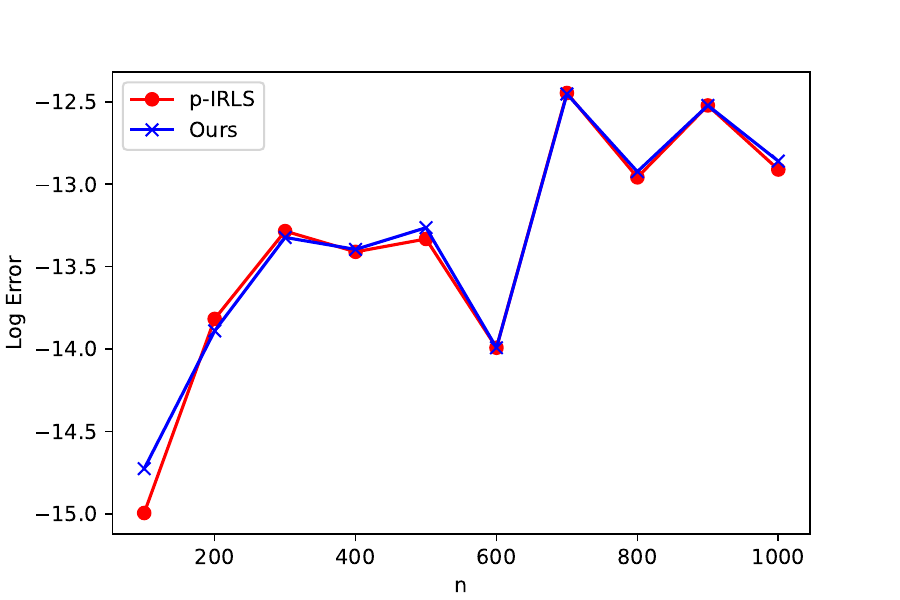}}\subfloat[matrix size=$500\times400$]{\includegraphics[width=0.25\textwidth]{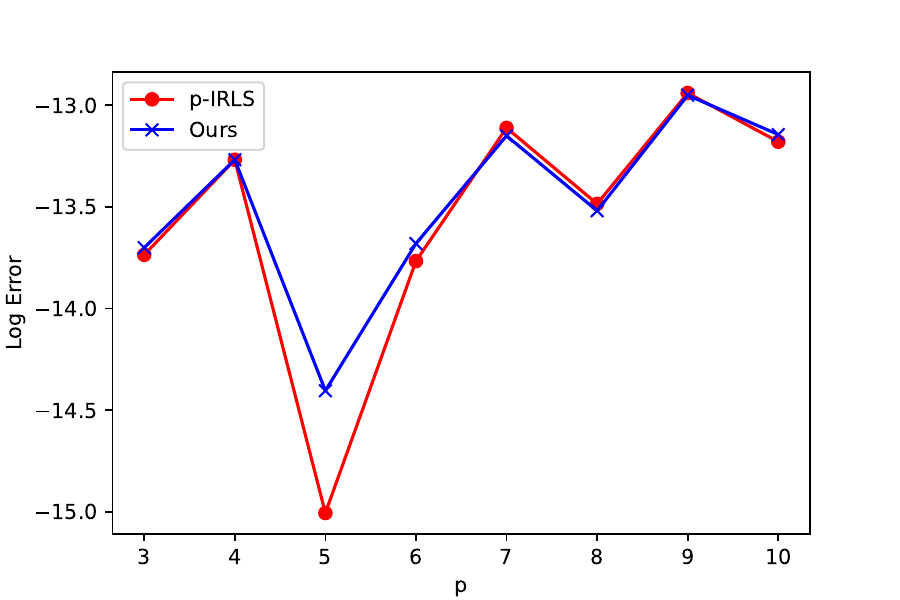}}\subfloat[Graph of $n$ nodes, $p=8$]{\includegraphics[width=0.25\textwidth]{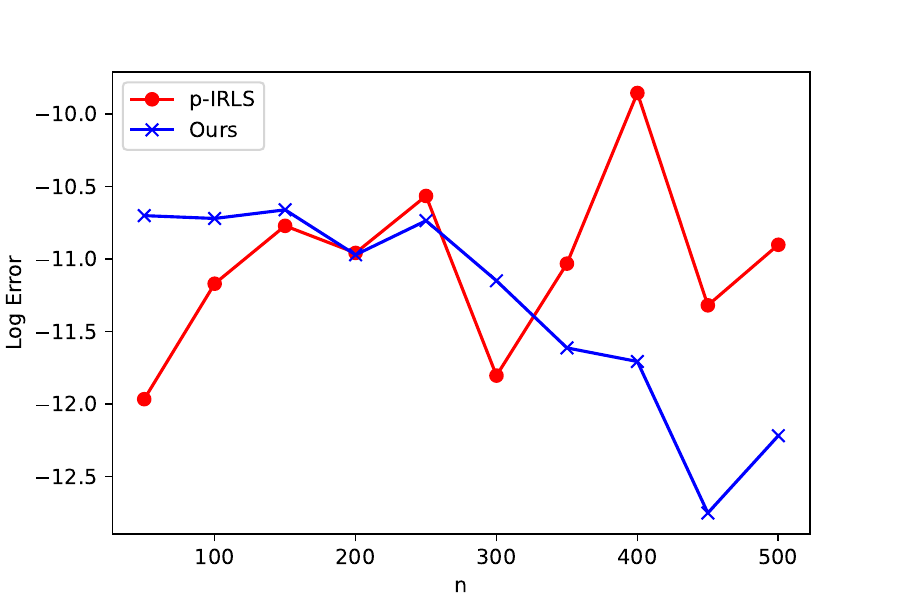}}\subfloat[Graph of $n=500$ nodes]{\includegraphics[width=0.25\textwidth]{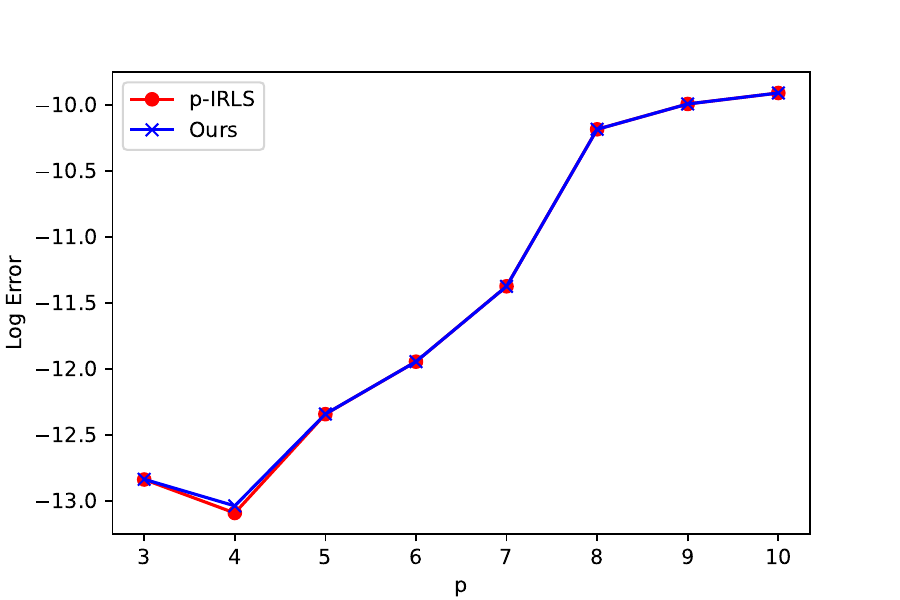}}

\caption{\protect\label{fig:error}Error of the solution against CVX/SDPT3
 solution in log10 scale.}
\end{figure*}

\subsection{When varying $\epsilon$}

In Figure \ref{fig:epsilon}, we plot iteration complexity and runtime
in seconds of our algorithm, $p$-IRLS and CVX when varying $\epsilon$.
Note that, CVX does not allow varying this parameter. In all experiment,
we fix $p=8$. For large instances, we only consider our solution
against $p$-IRLS.

\begin{figure*}
\subfloat[matrix size=$500\times400$]{\includegraphics[width=0.25\textwidth]{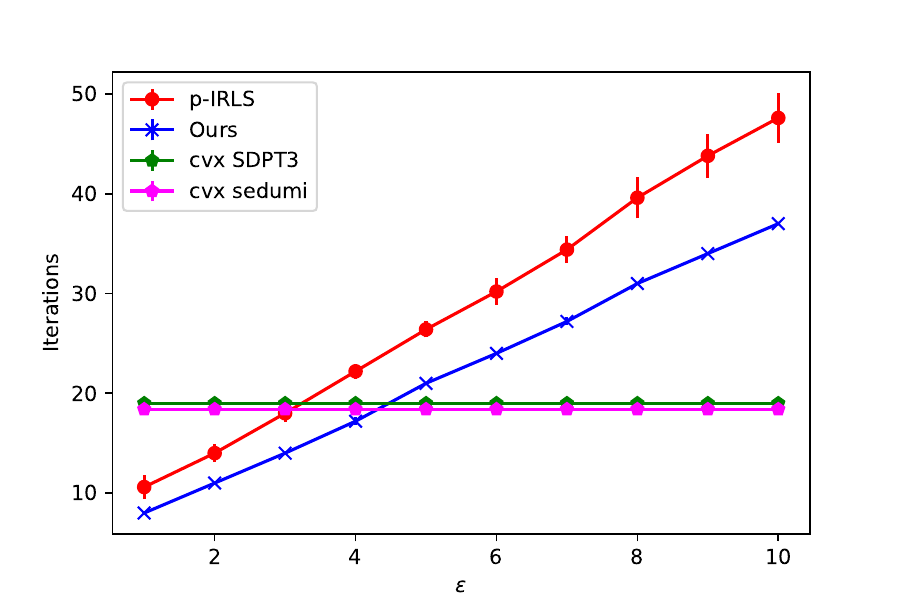}}\subfloat[matrix size=$500\times400$]{\includegraphics[width=0.25\textwidth]{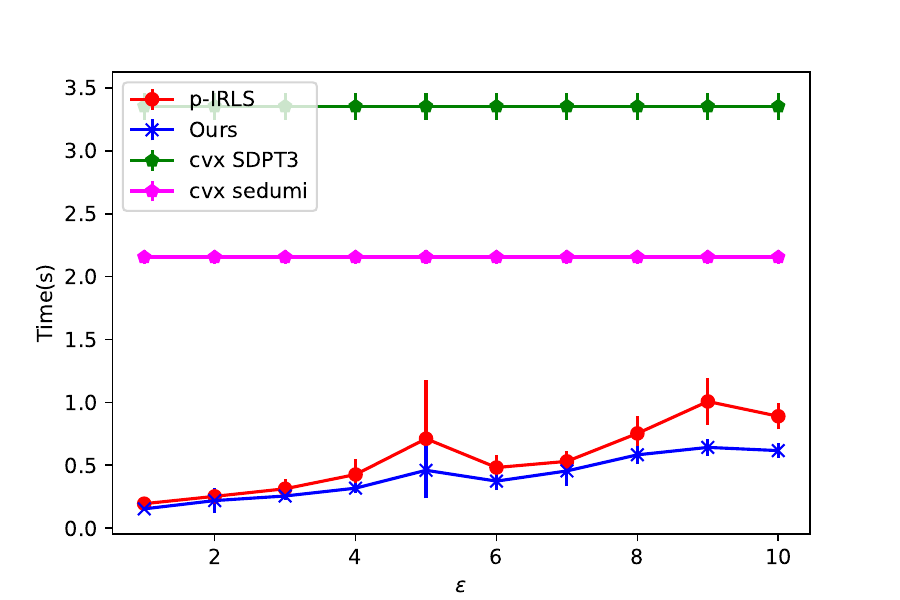}}\subfloat[Graph of $n=500$ nodes]{\includegraphics[width=0.25\textwidth]{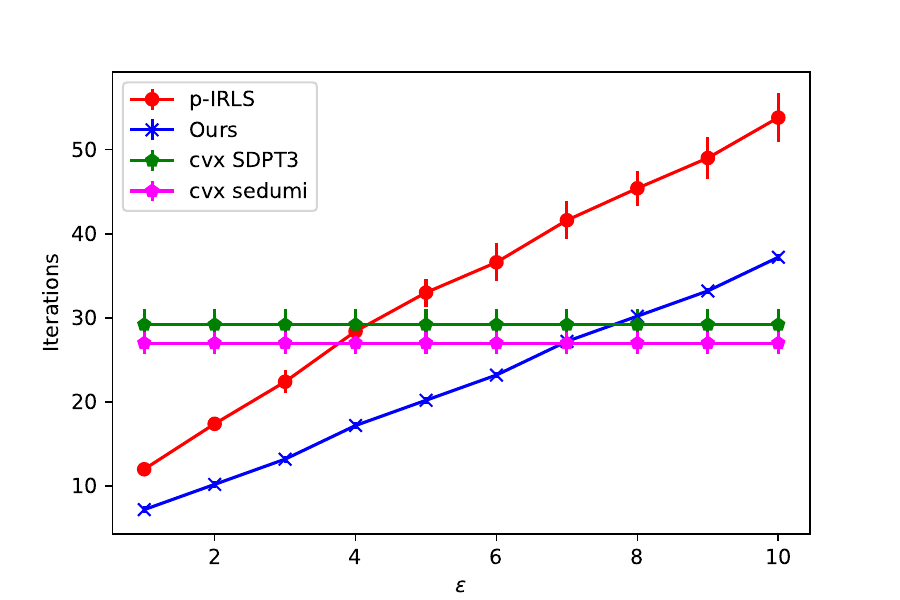}}\subfloat[Graph of $n=500$ nodes]{\includegraphics[width=0.25\textwidth]{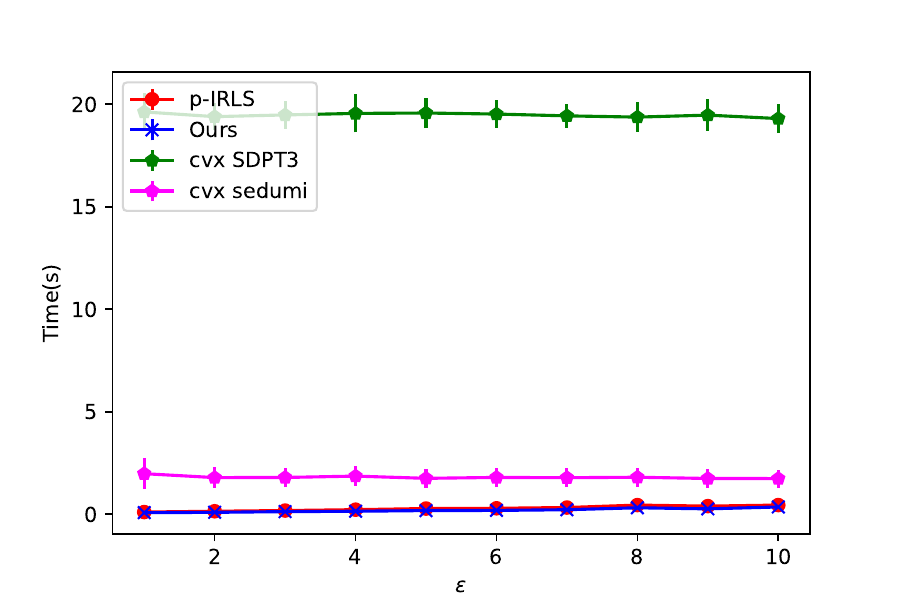}}

\subfloat[matrix size=$2500\times2400$]{\includegraphics[width=0.25\textwidth]{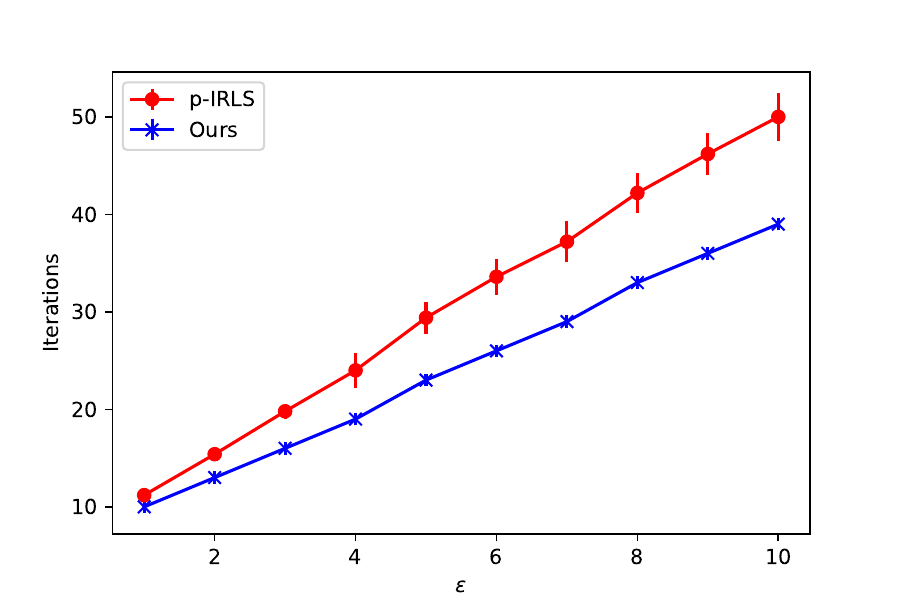}}\subfloat[matrix size=$2500\times2400$]{\includegraphics[width=0.25\textwidth]{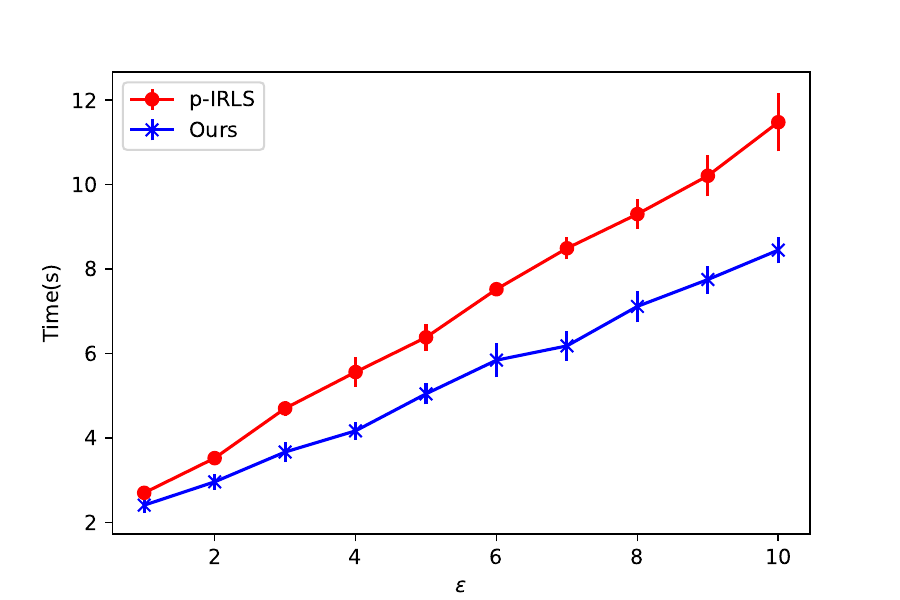}}\subfloat[Graph of $n=10000$ nodes]{\includegraphics[width=0.25\textwidth]{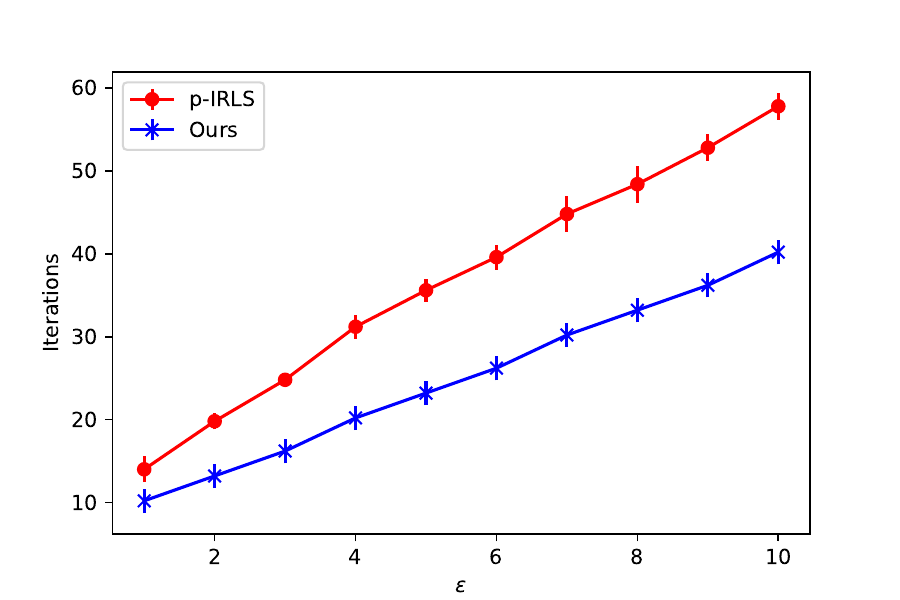}}\subfloat[Graph of $n=10000$ nodes]{\includegraphics[width=0.25\textwidth]{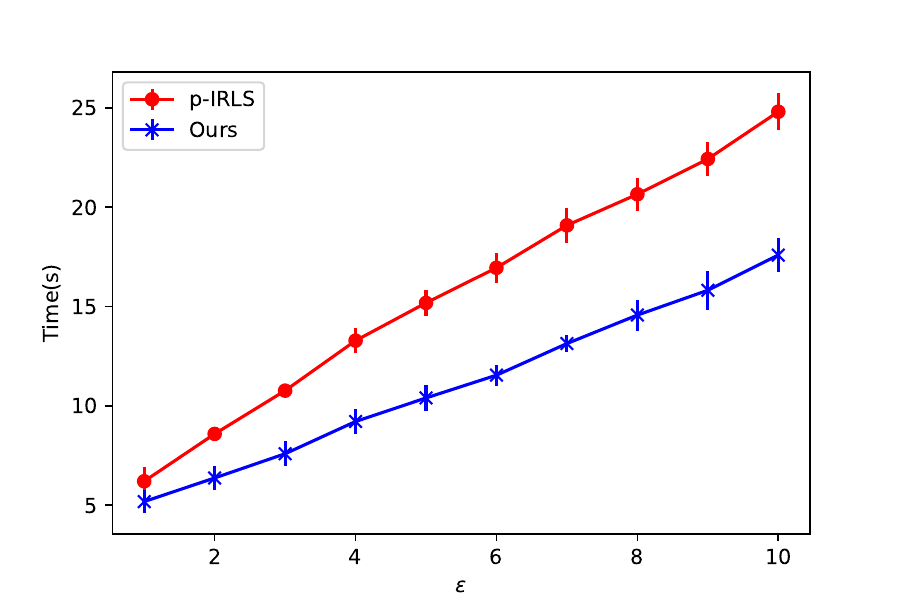}}

\caption{\protect\label{fig:epsilon}Performance when varying $\epsilon$ on
random matrices and random graphs instances.}
\end{figure*}

\subsection{For $1<p<2$}

In Figure \ref{fig:dual}, we plot iteration complexity and runtime
in seconds of our algorithm, $p$-IRLS and CVX on random matrices
of size $n\times(n-100)$. In all experiment, we fix $\epsilon=10^{-10}$.
We test with $p=1.1$ and $p=1.9$.

\begin{figure*}
\subfloat[$p=1.1$]{\includegraphics[width=0.25\textwidth]{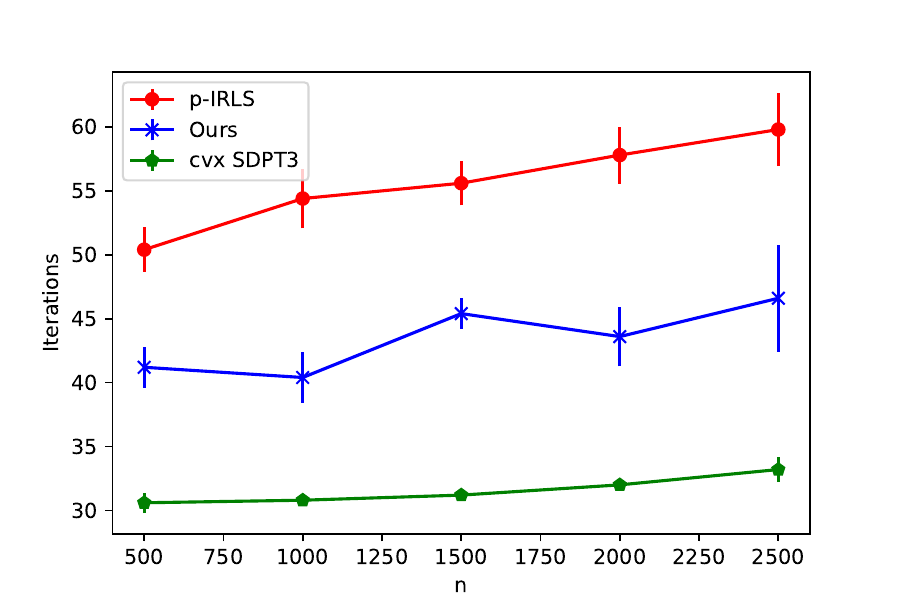}}\subfloat[$p=1.1$]{\includegraphics[width=0.25\textwidth]{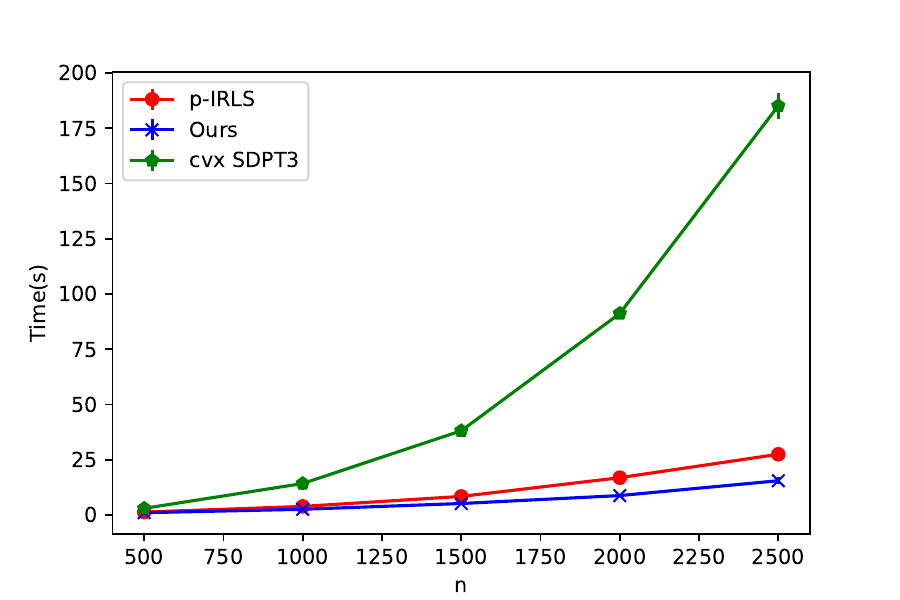}}\subfloat[$p=1.9$]{\includegraphics[width=0.25\textwidth]{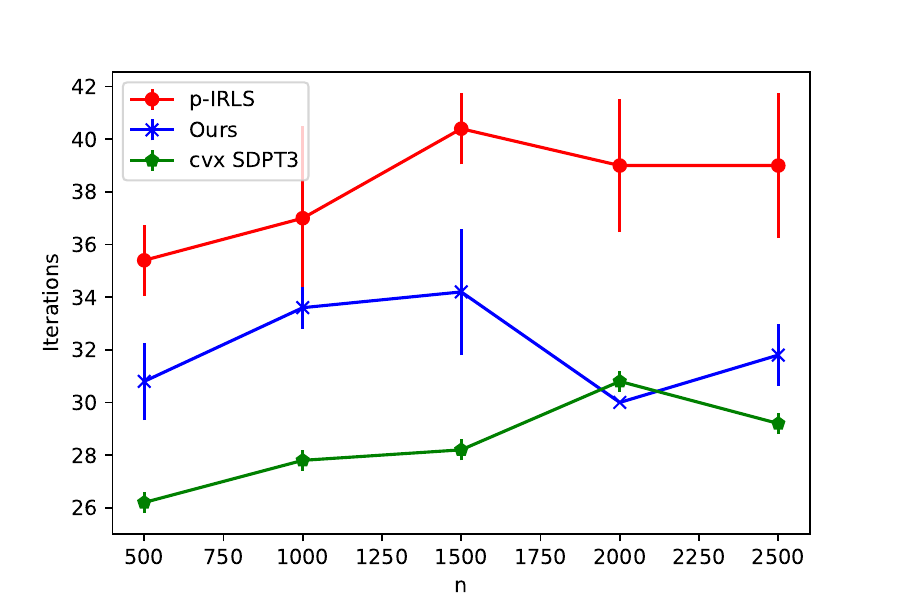}}\subfloat[$p=1.9$]{\includegraphics[width=0.25\textwidth]{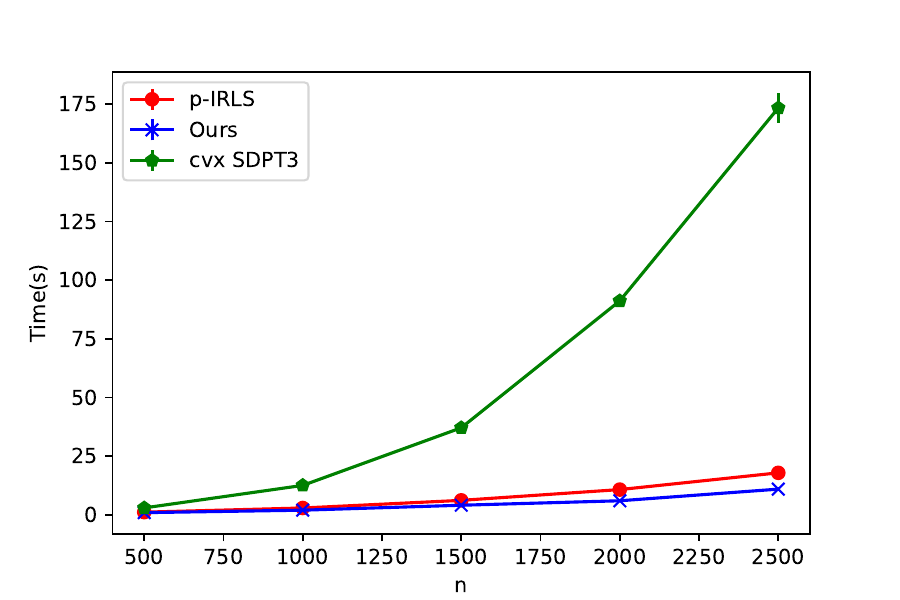}}

\caption{\protect\label{fig:dual}Performance when $p=1.1$ and $p=1.9$ on
random matrices of size $n\times(n-100)$.}
\end{figure*}

\newpage
\end{document}